\documentclass[a4paper,USenglish,cleveref]{lipics-v2021}

\hideLIPIcs
\nolinenumbers

\newcommand{\psfrage}[1]{{\color{blue}{\sf[PS: #1]}}}
\newcommand{\hpfrage}[1]{{\color{violet}\sf[HP: #1]}}
\newcommand{\fkfrage}[1]{{\color{teal}\sf[FK: #1]}}
\newcommand{\dbfrage}[1]{{\color{orange}\sf[DB: #1]}}
\renewcommand{\psfrage}[1]{} \renewcommand{\hpfrage}[1]{} \renewcommand{\fkfrage}[1]{} \renewcommand{\dbfrage}[1]{}

\usepackage{pgfplots}
\pgfplotsset{compat=newest}

\usepgfplotslibrary{groupplots}
\pgfplotsset{every axis/.style={scale only axis}}

\definecolor{veryLightGrey}{HTML}{F2F2F2}
\definecolor{colorSimdRecSplit}{HTML}{444444}
\definecolor{colorChd}{HTML}{377EB8}
\definecolor{colorPthashHem}{HTML}{A65628}
\definecolor{colorSicHash}{HTML}{4DAF4A}
\definecolor{colorPthash}{HTML}{984EA3}
\definecolor{colorRecSplit}{HTML}{FF7F00}
\definecolor{colorBbhash}{HTML}{F781BF}

\pgfplotscreateplotcyclelist{myColorList}{%
  colorSicHash,mark=o\\%
  colorRecSplit,mark=triangle\\%
  colorSimdRecSplit,mark=pentagon\\%
  colorChd,mark=square\\%
  colorPthash,mark=x\\%
  colorBbhash,mark=diamond\\%
}
\pgfplotscreateplotcyclelist{transparentHeatmap}{%
  colorSicHash,mark=*,fill=colorSicHash\\
}
\pgfdeclareplotmark{flippedTriangle}{%
  \pgfpathmoveto{\pgfqpointpolar{-90}{1.2\pgfplotmarksize}}%
  \pgfpathlineto{\pgfqpointpolar{30}{1.2\pgfplotmarksize}}%
  \pgfpathlineto{\pgfqpointpolar{150}{1.2\pgfplotmarksize}}%
  \pgfpathclose%
  \pgfusepath{stroke}
}

\pgfplotsset{
  mark repeat*/.style={
    scatter,
    scatter src=x,
    scatter/@pre marker code/.code={
      \pgfmathtruncatemacro\usemark{
        or(mod(\coordindex,#1)==0, (\coordindex==(\numcoords-1))
      }
      \ifnum\usemark=0
        \pgfplotsset{mark=none}
      \fi
    },
    scatter/@post marker code/.code={}
  },
  major grid style={thin,dotted},
  minor grid style={thin,dotted},
  ymajorgrids,
  yminorgrids,
  every axis/.append style={
    line width=0.7pt,
    tick style={
      line cap=round,
      thin,
      major tick length=4pt,
      minor tick length=2pt,
    },
    mark options={solid},
  },
  legend cell align=left,
  legend style={
    line width=0.7pt,
    /tikz/every even column/.append style={column sep=3mm,black},
    /tikz/every odd column/.append style={black},
    mark options={solid},
  },
  legend style={font=\small},
  title style={yshift=-2pt},
  enlarge x limits=0.04,
  every tick label/.append style={font=\footnotesize},
  every axis label/.append style={font=\small},
  every axis y label/.append style={yshift=-1ex},
  /pgf/number format/1000 sep={},
  axis lines*=left,
  xlabel near ticks,
  ylabel near ticks,
  axis lines*=left,
  label style={font=\footnotesize},
  tick label style={font=\footnotesize},
  cycle list name=myColorList,
  plotParameters/.style={
    width=22.0mm,
    height=30.0mm,
    ymax=5e7,
    ymin=1e4,
  },
  plotScaling/.style={
    width=38.0mm,
    height=30.0mm,
  },
  plotScalingConfigs/.style={
    width=35.0mm,
    height=30.0mm,
  },
  plotPareto/.style={
    width=25.0mm,
    height=30.0mm,
  },
  plotHfEvals/.style={
    width=35.0mm,
    height=20.0mm,
  },
  plotLeafMethods/.style={
    width=38.0mm,
    height=30.0mm,
  },
  plotProbabilities/.style={
    width=35.0mm,
    height=24.0mm,
    only marks,
    mark size=.75pt,
    cycle list name=transparentHeatmap,
  },
}

\usepackage{dsfont}
\usepackage{mathbbol}
\usepackage{pgfplots}
\usepackage{caption}
\usepackage{subcaption}
\usepackage{hyperref}
\usepackage{booktabs}
\usepackage[nocompress]{cite}
\crefname{listing}{Algorithm}{Algorithms}

\newcommand{\mytitle}{High Performance Construction of RecSplit Based Minimal Perfect Hash Functions}
\title{\mytitle}
\titlerunning{\mytitle}

\author{Dominik Bez}{Karlsruhe Institute of Technology, Germany}{dominik.bez@student.kit.edu}{}{}
\author{Florian Kurpicz}{Karlsruhe Institute of Technology, Germany}{kurpicz@kit.edu}{https://orcid.org/0000-0002-2379-9455}{}
\author{Hans-Peter Lehmann}{Karlsruhe Institute of Technology, Germany}{hans-peter.lehmann@kit.edu}{https://orcid.org/0000-0002-0474-1805}{}
\author{Peter Sanders}{Karlsruhe Institute of Technology, Germany}{sanders@kit.edu}{https://orcid.org/0000-0003-3330-9349}{}

\newcommand{\myauthorrunning}{D. Bez, F. Kurpicz, H.-P. Lehmann, P. Sanders}
\authorrunning{\myauthorrunning}
\Copyright{Dominik Bez, Florian Kurpicz, Hans-Peter Lehmann, and Peter Sanders}
\hypersetup{
  colorlinks=true,
  pdftitle={\mytitle},
  pdfauthor={\myauthorrunning},
  pdfsubject={}
}


\supplementdetails[subcategory={Source code}]{Software}{https://github.com/ByteHamster/GpuRecSplit}
\supplementdetails[subcategory={Comparison with competitors}]{Software}{https://github.com/ByteHamster/MPHF-Experiments}

\acknowledgements{
This paper is based on and has text overlaps with Dominik Bez' Master's thesis~\cite{bez2022gpurecsplit}.
We refer readers to that thesis for a detailed evaluation of the effects of low-level decisions like the choice of different similar SIMD instructions.
}
\funding{
\flag{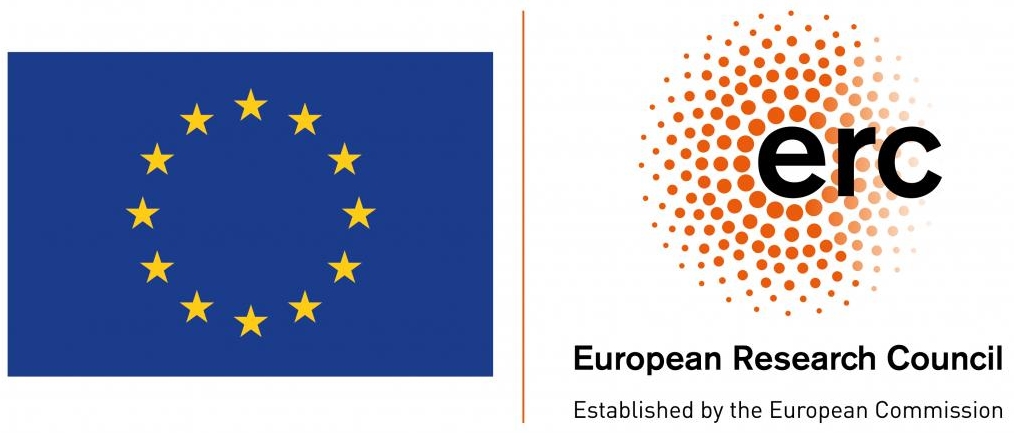}
This project has received funding from the European Research Council (ERC) under the European Union’s Horizon 2020 research and innovation programme (grant agreement No. 882500).
}

\begin{document}
\maketitle

\def\GpuSpeedup{5438}

\def\SimdSpeedup{239}

\begin{abstract}
  A minimal perfect hash function (MPHF) bijectively maps a set $S$ of objects to the first $|S|$ integers.
  It can be used as a building block in databases and data compression.
  RecSplit [Esposito~et~al.,~ALENEX'20] is currently the most space efficient practical minimal perfect hash function.
  It heavily relies on trying out hash functions in a brute force way.

  We introduce \emph{rotation fitting}, a new technique that makes the search more efficient by drastically reducing the number of tried hash functions.
  Additionally, we greatly improve the construction time of RecSplit by harnessing parallelism on the level of bits, vectors, cores, and GPUs.

  In combination, the resulting improvements yield
  speedups up to \SimdSpeedup{} on an 8-core CPU
  and up to \GpuSpeedup{} using a GPU.  The
  original single-threaded RecSplit implementation needs 1.5 hours
  to construct an MPHF for $5$~Million
  objects with 1.56 bits per object.  On the GPU,
  we achieve the same space usage in just 5
  seconds.  Given that the speedups are larger
  than the increase in energy consumption, our
  implementation is more energy efficient than the
  original implementation. %

\end{abstract}
\newpage

\section{Introduction}
A \emph{Perfect Hash Function} (PHF) is a hash function that does not have collisions, i.e., is injective, on a given set $S$ of objects.
Evaluating the PHF on any object not in $S$ can return an arbitrary value.
A \emph{Minimal Perfect Hash Function} (MPHF) maps the objects in $S$ to the first $n=|S|$ integers, so it is bijective.
MPHFs are useful in many applications, for example, to implement hash tables with guaranteed constant access time \cite{fredman1984storing}.
By storing only fingerprints in the hash table cells \cite{fan2014cuckoo,bender2018bloom}, we obtain an \emph{approximate membership data structure}.
Storing payload data in the cells, we obtain an updatable retrieval data structure~\cite{muller2014retrieval}.
Finally, the perfect hash function values can be used as small identifiers of the input objects \cite{botelho2007perfect}, which are easier to handle and more space efficient than, for example, strings.

MPHFs can be very compact -- the theoretically minimal space usage is 1.44 bits per object \cite{belazzougui2009hash}.
Currently, the most space-efficient practical MPHF is RecSplit \cite{esposito2020recsplit}.
It provides various trade-offs between the space consumption, construction time, and query time.

In this paper, we provide several improvements inside the RecSplit framework.
We first describe RecSplit and other preliminaries in \cref{s:prelim} and briefly review related work in \cref{s:related}.
As a core step during construction, RecSplit tries out hash functions on a small set of objects until one hash function is a bijection.
We introduce a new bijection search mechanism in \cref{s:bijections}, which reduces the search space of the brute force algorithm compared to the original method.
\emph{Rotation fitting} hashes the objects to two sets and tries to fit one set into the ``holes'' of the other set by rotating (cyclically shifting) it.
As a positive side effect, this approach makes good use of bit parallelism.

We then parallelize RecSplit (with and without rotation fitting) using the vector parallelism available with \emph{Single Instruction Multiple Data} (SIMD) instructions and the thread parallelism available with multicore CPUs and GPUs.
Given that hash function construction here is mostly compute bound and can be done in parallel for a huge number of small subproblems, the GPU is an ideal hardware.
Utilizing GPUs for evaluating hash functions is known from mining of cryptocurrencies with proof-of-work approach (e.g., Bitcoin).
Our extensive evaluation in \cref{s:experiments} shows speedups of up to 50 using SIMD, \SimdSpeedup{} when additionally using multi-threading with 16 threads, and \GpuSpeedup{} using a GPU, compared to the original single-threaded implementation without rotation fitting.
Because GPUs are so much faster at constructing MPHFs, they lead to a better energy efficiency than the CPU, as we show in the experiments.
Finally, in \cref{s:conclusion}, we summarize the results and give directions for future research.

\subparagraph*{Our Contributions.}
With \emph{rotation fitting}, we introduce a new method for searching for bijections that can be used in RecSplit.
We significantly accelerate the construction by four kinds of parallelism (bits, vectors, multicores, and GPU).
Together, this accelerates RecSplit constructions by a factor up to \GpuSpeedup{} and even makes its construction
performance competitive to significantly less space efficient minimal perfect hash functions.

\section{Preliminaries}\label{s:prelim}
In \cref{s:basics}, we first shortly describe basic techniques needed by our implementation.
We then continue with describing RecSplit in detail in \cref{s:recsplit}.
Finally, we describe SIMD in \cref{s:simd} and GPUs in \cref{s:gpus}.

\subsection{Basics} \label{s:basics}

\subparagraph*{Words and Bit Vectors.}
An important operation in RecSplit is \texttt{popcount}, which returns the number of 1-bits in a word.
Given a bit vector, the $\textit{rank}_1(x)$ operation returns the number of 1-bits before position $x$, and the $\textit{select}_1(x)$ operation returns the position of the $x$-th 1-bit.
The operation can be executed in constant time \cite{clark1997compact,jacobson1989space} and has very fast and space-efficient implementations \cite{vigna2008broadword,kurpicz2022pasta}.
An additional operation we need in this paper is $\mathrm{rot}_k^{i}(x)$
which rotates (i.e., cyclically shifts) the $k$ least significant bits of $x$ by $i$ bit positions. This can be implemented in a bit parallel way using shifting and masking.

\subparagraph*{Golomb-Rice.}\label{s:golombRice}
The Golomb code \cite{golomb1966run} is a variable length code that is optimal for geometric distributions.
Golomb-Rice \cite{rice1979some} is a faster special case, which is almost as space efficient.
Given a parameter $\tau$ and the number $x$ to store, the $\tau$ least significant bits of $x$
form the \emph{fixed part} which is stored directly.
The remaining bits are encoded in unary, consisting of $\lfloor x/2^{\tau}\rfloor$ 0-bits and a final 1-bit.
To access one element, we can get the lower bits from the array of fixed parts and the upper bits through two $\textit{select}_1$ queries.

\subparagraph*{Elias-Fano.}\label{s:eliasFano}
An Elias-Fano representation~\cite{Elias74,Fano71} can be used to store a monotonic sequence of integers \(p_1,\dots,p_k\) with \(p_k\leq U\).
Similar to Golomb-Rice codes, the least significant bits of each value are stored directly in the lower-bits array and can be accessed directly.
The remaining most significant bits $u$ at index $i$ are encoded as a 1-bit in a bit vector at position $i+u$.
This means that by executing a $\textit{select}_1$ query on the upper bits and looking up the lower bits, we can restore any value in constant time.
Using this representation, the sequence can be stored using \(k(2+\log(U/k))\) bits.

\begin{figure}[t]
  \centering
  \includegraphics[scale=0.7]{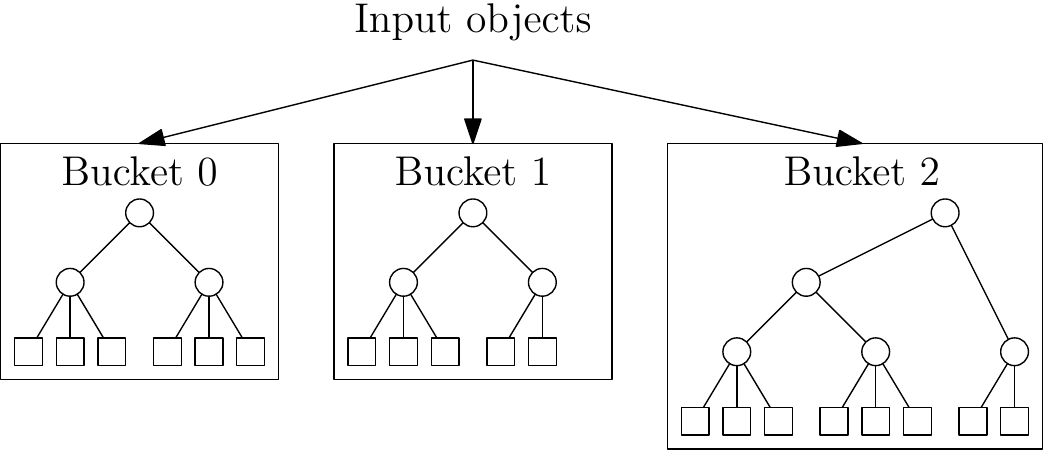}
  \caption{\label{fig:overview} Illustration of the overall RecSplit data structure. Circular nodes of the trees represent splittings, squares represent bijections.}
\end{figure}

\subsection{RecSplit} \label{s:recsplit}

We now describe RecSplit \cite{esposito2020recsplit}, the MPHF that this paper is based on.
\Cref{fig:overview} illustrates the overall data structure.
The first step of the construction is to apply an initial hash function on every object of the input to generate objects of uniform distribution.
These objects are mapped to different buckets of expected size $b$, where $b$ is a tuning parameter.

\subparagraph*{Splitting Trees.}\label{s:splittingTrees}
In each bucket, RecSplit constructs an independent \emph{splitting tree}.
The tree partitions the objects into smaller and smaller sets until the individual sets have a small configurable size $\ell$.
The splitting tree has a well-defined shape, depending only on the leaf size $\ell$ and the number of objects in the bucket.
At each inner node, RecSplit tries random hash functions to find one that distributes the objects to the child nodes according to the tree structure.
The number of child nodes of an inner node is called fanout.
The fanout is optimized in such a way that the expected amount of work to find the splitting is roughly equal to the amount of work in all children combined.
The fanouts of the two bottom-most levels are $\max \{2, \lceil 0.35 \ell + 0.55\rceil\}$ and $\max\{2, \lceil 0.21 \ell + 0.9\rceil\}$.
In the terminology of the RecSplit paper, these levels are called \emph{lower aggregation levels}.
The levels above, also called \emph{upper aggregation levels}, simply use a fanout of $2$.

\subparagraph*{Bijections.}\label{s:bijectionsExplanation}
The lowest level of the splitting tree is called \emph{leaf level}.
Each leaf, except for possibly the last, contains $\ell$ objects.
This is small enough that it is feasible to search for a bijective mapping by trying random hash functions using brute force.
The inner loop of the bijection search applies a hash function modulo $\ell$ on each object.
It converts the value to a bit by taking two to the power of it, and sets the corresponding bit in a bit vector of length $\ell$ using a logical \texttt{OR} operation.
After hashing all objects, if the resulting bit vector has all its bits set to $1$, it means that the hash function is a bijection on the leaf.
If it is not, RecSplit tries the next hash function.

\subparagraph*{Representation.}\label{s:storage}
Because the splitting trees have a well-defined shape, it is enough to store the hash function identifier at each node in preorder.
These numbers are encoded with Golomb-Rice code, where all unary parts and all binary parts of a tree are stored together.
The optimal Golomb parameter $\tau$ is different based on the layer in the tree and can be pre-calculated.
The encodings of all splitting trees from all buckets are concatenated in a single bit vector.
An additional sequence with encoding based on Elias-Fano encodes both the prefix sums of the number of objects in each bucket and the positions where the encoding of each bucket starts.

\subparagraph*{Query.}
A RecSplit hash function can be evaluated by determining the bucket of an object and locating its encoding.
The splitting tree in the bucket is traversed from the root to a leaf by applying the splitting hash function, which determines the child to descend into.
Finding the encoding of a subtree is possible by executing a $\textit{select}_1$ query on the upper bits of the Golomb-Rice coded hash function identifiers.
During traversal, the number of objects stored in children left to the one descended into are accumulated.
The final hash value is then the sum of the value of leaf bijection, the number of objects to the left in the splitting tree, and the total size of previous buckets.

The combination of brute force splitting and bijections is highly space efficient from an information-theoretical point of view -- disregarding overheads due to encoding and metadata, optimal space consumption can be achieved.
Consequently, as the leaf size $\ell$ gets larger, optimal space is approached \cite{esposito2020recsplit}.

\subsection{SIMD} \label{s:simd}
It is common, especially in perfect hashing, that the same operation needs to be executed on different data.
This can be achieved with a simple loop, which means that the corresponding instructions must be decoded by the hardware for every element.
This can be improved by using \emph{Single Instruction, Multiple Data} (SIMD) \cite{flynn1972some}.
A single instruction is used to apply the same operation on a \emph{vector} of several elements.
We refer to a single element within a SIMD vector as a \emph{lane}.
For example, a vector may contain 16 lanes with 32 bits each, i.e., the vector contains 512 bits overall.
The exact set of operations depends on the concrete implementation of the SIMD model.
On many Intel and AMD processors, SIMD operations are available through the Advanced Vector Extensions (AVX) \cite{intel2011avx}.
AVX-512 \cite{intel2013avx512} extends these operations to 512-bit vectors and is divided into many smaller subsets that offer additional operations.
A subset that is useful for our implementation is AVX512VPOPCNTDQ, which provides \texttt{popcount} on 512-bit vectors with lanes of size 32 and 64 bits.
The $\mathrm{rot}_k^{i}$ function that cyclically shifts bits (see \cref{s:basics}) can be implemented in parallel using SIMD.

\subsection{GPUs} \label{s:gpus}
Graphics Processing Units (GPUs) are specialized processors initially designed for computer graphics applications.
Over the last decades, GPUs evolved to general purpose processors for highly parallelizable tasks.
We now describe the hardware and programming interface in the following paragraphs.
To provide a grasp of the dimensions of a current GPU, we give metrics of the NVIDIA RTX 3090 \cite{nvidia2020ampere}, which is also used for our experiments (see \cref{s:experiments}).

\subparagraph*{Compute Hardware.}
A GPU consists of several streaming multiprocessors (SMs) (RTX 3090: 82).
Each SM contains many arithmetic logic units (ALUs) to perform computations (RTX 3090: 64 integer ALUs).
Several threads (RTX 3090: 32) operate in \emph{lock-step}, i.e., they execute the same instruction at the same time.
Such a bundle of threads is called \emph{warp}.
Threads are masked out for instructions they should not execute.
This means that in loops, each thread in a warp has to iterate as many times as the thread with the largest number of iterations.
To hide latencies, e.g., for memory access, each SM is oversubscribed with more threads than ALUs, and the GPU schedules the threads efficiently.
Multiple warps of threads form a \emph{thread block}.
Thread blocks are guaranteed to reside on the same SM, which enables them to cooperate.%

\subparagraph*{Memory.}
The \emph{global memory} is the largest and slowest memory on the GPU (RTX 3090: 24~GB).
When multiple threads of a warp access the memory simultaneously, the hardware serves the requests with as few memory transactions as possible.
\emph{Shared memory} is a fast memory placed on each SM.
It is shared between the threads of the same thread block.
On the RTX 3090, shared memory and L1 cache are allocated on the same memory areas.
The data in shared memory is partitioned into 32 memory banks, and the $i$-th 32-bit word is stored in bank $i \textrm{ mod } 32$.
When multiple threads simultaneously access different words within the same bank, the access operations have to be serialized.

\subparagraph*{CUDA.}
An efficient way to develop applications on NVIDIA GPUs is CUDA \cite{nvidia2022cuda}.
Functions which can be executed on the GPU are called \emph{kernels}.
Each kernel is executed on a \emph{grid} of \emph{thread blocks}.
The grid size and the number of threads per block can be selected by the user.
The user can create several \emph{streams}.
The kernels and data transfers launched into a specific stream are executed in order, but operations in different streams can arbitrarily overlap.

\section{Related Work}\label{s:related}
Perfect Hashing is an active area of research \cite{czech1992optimal, weaver2020constructing, belazzougui2009hash, lehmann2022sichash, pibiri2021pthash, limasset2017fast, botelho2013practical, fox1992faster, botelho2007simple, CSLSR11, MSSZ14}.
Due to a lack of space, we only describe the most recent and fastest algorithms here.
For a more detailed overview of recent methods, refer to Ref. \cite{lehmann2022sichash}.
To the best of our knowledge, there is no technique that constructs MPHFs on the GPU yet.
Lefebvre and Hoppe \cite{lefebvre2006perfect} describe the GPU evaluation of MPHFs that were constructed on CPUs.

\subparagraph*{FiPHa/BBHash.}%
A fast and simple approach to minimal perfect hashing uses fingerprinting and bumping \cite{CSLSR11,MSSZ14,limasset2017fast}.
BBHash \cite{limasset2017fast} is a publicly available parallel implementation.
The set $S$ of input objects is hashed using a
hash function $h\rightarrow \beta n$ for a tuning
parameter $\beta$.  The set $S'$ of objects that
have a collision is handled recursively.
Consider
the bit vector $b$ with
$b[i]=1$ iff $|\{s\in S:h(s)=i\}|=1$.
Then $\textit{rank}_1(h(s))$ defines an MPHF on $S\setminus S'$.
This approach needs at least $e$
bits per object (when $\beta=1$) and provides
efficient queries when about 4 or more bits per
object are available (using larger values of
$\beta$). An advantage is the very simple and easily
parallelizable construction.

\subparagraph*{PTHash.}
PTHash \cite{pibiri2021pthash} is based on FCH \cite{fox1992faster} which can be considered a predecessor of the hash-and-displace technique \cite{belazzougui2009hash}.
The objects are first distributed into different buckets using a hash function, but the distribution is not uniform.
Specifically, about 60\% of the objects are mapped to 30\% of the buckets.
The buckets are then processed in order of decreasing size.
For each bucket, a hash function is searched such that each object can be placed in the output domain without colliding with other objects that are already placed.
The hash function identifiers are searched linearly and then stored in compressed form with several possible compression schemes.
The proclaimed goal of PTHash is fast query times.
Using an appropriate compression scheme, only a single memory access is required to find the hash value, and the remaining operations are simple hash function evaluations and arithmetic.
Compared to the original implementation of RecSplit, PTHash consumes more space, but has faster queries and faster construction time.
PTHash-HEM \cite{pibiri2021parallel} is an implementation that first partitions the input and then constructs each partition independently in parallel.

\subparagraph*{SicHash.}
SicHash \cite{lehmann2022sichash} is based on the simple idea to store the index of the hash function to be used in a retrieval data structure.
It can capitalize on recent progress on fast and nearly space optimal retrieval \cite{DHSW22}.
Computing a valid index for all objects amounts to constructing a cuckoo hash table \cite{pagh2004cuckoo,FPSS05}.
In contrast to the brute force methods at the core of PTHash and RecSplit, this can be done in near linear time even on large tables.
SicHash refines this basic approach using a mix of several fixed precision retrieval data structures and by using many small(ish) cuckoo hash tables rather than a single large table.
Roughly, SicHash allows faster construction than PTHash while offering similar query time and space consumption.

\section{Rotation Fitting}\label{s:bijections}
The general idea of RecSplit consists of two independent steps, bijections and splittings (see \cref{s:recsplit}).
In this section, we introduce a new method for searching for bijections in RecSplit's leaf nodes.
As a reminder, given $m$ objects, we are looking for a way to quickly find a mapping of the objects to the first $m$ integers without any collisions.
The original implementation tries out hash functions using brute force until one of them is a bijection.

\emph{Rotation fitting} ensures that we need significantly fewer hash function evaluations.
From the result of one evaluation, we derive additional candidates that are very fast to compute.
Rotation fitting is efficient when $m \leq w$, where $w$ is the size of a machine word.
We randomly distribute the objects into two sets $A$ and $B$ by using a 1-bit hash function.
The 1-bit hash function is the same for all leaf nodes and does not ensure that $A$ and $B$ have the same size.
Now we search for a hash function $h$ that gives a bijection on the leaf.
Like in the original RecSplit implementation, we calculate the hash value of all objects in $A$ and set the respective bits in the word $a$ to $1$.
The function $h$ may be ruled out as a valid bijection by calculating the \texttt{popcount} of $a$.
Analogously, the set $B$ is mapped to the word $b$ using the same hash function $h$.
Let us now rotate (i.e., cyclically shift) the bits in $b$.
If we can find a rotation value such that the 1-bits in $b$ fit exactly onto the 0-bits in $a$, we have found a bijection on the leaf.
More formally, this is the case if there is an $r \in \{0,...,m-1\}$, such that $a\, |\, \mathrm{rot}^r_m(b)$ has the $m$ least significant bits all set.
In \cref{s:proofRotations}, we show that for large \(m\) the probability of finding a bijection using rotation fitting is about \(m\) times higher than the probability when using RecSplit's brute force approach.

To efficiently store $r$, we only try hash function identifiers which are multiples of $m$. %
This number plus $r$ is stored for each leaf.
We can restore $r$ later by calculating modulo $m$ and restore the hash function index by rounding down to the next multiple of $m$.
At query time, a rotation corresponds to an addition modulo $m$ to each object in the set $B$.
The space overhead per object introduced by rotation fitting tends to 0 for large \(m\) (see \cref{s:proofRotations}).

\subparagraph*{Lookup Tables.}
It is possible to avoid trying out all $m$ rotations by using a lookup table $t$.
For all possible values of $a$, this table contains a rotation parameter $t[a]$ such that $\mathrm{rot}_m^{t[a]}(a)$ is minimal.
If a value $x$ can be rotated to get the value $y$, then $\mathrm{rot}_m^{t[x]}(x) = \mathrm{rot}_m^{t[y]}(y)$.
Let $c=2^m-1$ be the word where the $m$ least significant bits are set.
The value $\hat b = b \oplus c$ is $b$ with the $m$ least significant bits flipped.
Note that $b$ can fill the holes in $a$ if and only if $\hat b$ can be rotated to match $a$.
Thus, the necessary rotation of $b$ can be calculated as $r=(t[\hat b] - t[a]) \bmod m$ using two table lookups.
Rotation $r$ is valid if $a | \mathrm{rot}_m^r(b) = c$.

Because rotation is a very cheap operation, preliminary experiments show no improvement by lookup tables.
Especially on GPUs, shared memory is a scarce resource and global memory is too slow.
Our implementation therefore does not use lookup tables.
Nonetheless, rotation fitting with lookup tables provides an asymptotic improvement of the running time by a factor of $m$.
We also find the idea to normalize random permutations like this an interesting and novel concept.
Applying this idea to other permutations is left for future research.

\section{Parallelization}\label{s:parallel}
We describe the SIMD implementation in \cref{s:parallelSimd} and, on top of it, the multi-threaded implementation in \cref{s:parallelMt}.
Finally, we describe the GPU implementation in \cref{s:parallelGpu}.

\subsection{SIMD} \label{s:parallelSimd}
For the SIMD parallelization, we focus on the description of bijections and splittings, which (in most configurations) take most time of the construction.
While we additionally accelerate the construction of the Elias-Fano data structure, the ideas are more straight forward and are omitted due to space constraints.
The main idea of our SIMD parallelization is to try multiple hash function seeds simultaneously.
Depending on the operation, we use SIMD lanes with a width of either 32 bits or 64 bits.

\subparagraph*{Bijections.}
For the bijections, each SIMD lane is responsible for trying one hash function.
For this, we load consecutive hash function identifiers and the same input object to each lane of a SIMD vector, and evaluate the hash function.
The resulting hash value in each lane is converted to a single bit by taking two to the power of it.
After calculating the logical \texttt{OR} of these bits for all objects in the set, we check for a bijection by comparing each lane with a constant that has all $m$ lower bits set to $1$.
For rotation fitting, remember that the number we store as a seed is the hash function identification plus the rotation.
This number should be as small as possible to avoid wasting space, so caution must be taken when trying out the rotations.
If one lane finds a bijection, it might be possible that a higher rotation leads to a bijection on a lane with a smaller hash function index.
Because this gives a smaller overall number to store, we always try all rotation values, even if a bijection is found.

\subparagraph*{Splittings.}
For the splittings, the original implementation uses small arrays of counters.
Each counter contains the number of objects hashed to the respective split section.
Instead, we use two different methods.
For the \emph{upper} aggregation levels with fanout $2$, we use a single counter for the number of objects hashed to the left child.
The number of objects in the right child can then be determined by subtraction.
For all practical leaf sizes ($\ell \leq 24$), each counter of a valid \emph{lower} level splitting fits into a single byte.
Because an overflowing counter for one child would then just add $1$ to the next counter, such overflows cannot make an invalid splitting look valid.
When a seed for a valid splitting is found, we need to redistribute the objects.
We now use SIMD to apply the same hash function to several objects at once, and store the results in an array.
We then redistribute the objects without SIMD parallelism.

\subsection{Multi-Threading} \label{s:parallelMt}
The original RecSplit implementation only uses a single thread.
This leaves a lot of processing power unused since most modern processors contain multiple processing cores.
As stated in the original RecSplit paper \cite{esposito2020recsplit}, parallelizing RecSplit is fairly easy because the buckets are completely independent of each other.
First, we sort the input objects by their bucket index in parallel, and then determine the bucket borders.
We then start several threads and assign a consecutive portion of the buckets to each thread.
Because the number of buckets is large and the input objects are hashed to buckets uniformly, the load of all threads is reasonably balanced.

After a splitting or bijection is found, it must be stored in the Golomb-Rice coded sequence.
To avoid synchronization, each thread uses its own local sequence and treats its input as if it was the complete input.
This means it also stores the pointers to the start of each bucket encoding locally.
After all threads are done, we sequentially concatenate the Golomb-Rice sequences and build the combined Elias-Fano data structure holding the prefix sum of bucket sizes and pointers to the bucket encodings.

\subsection{GPU} \label{s:parallelGpu}
In the GPU implementation, we first partition the objects to their buckets and partition the buckets by their respective size.
We then use the GPU to determine the splittings and bijections within the buckets.
Buckets with the same size have splitting trees with the same shape and can therefore be handled efficiently within the same set of kernel calls.
This keeps the number of kernel calls small and is important for scalability.
Using CUDA's streams, we additionally construct different bucket shapes concurrently, to utilize the GPU in case the number of buckets having a specific shape is small.
For an overview, see \cref{fig:gpuTree}.

\begin{figure}[t]
  \centering
  \includegraphics[scale=0.6]{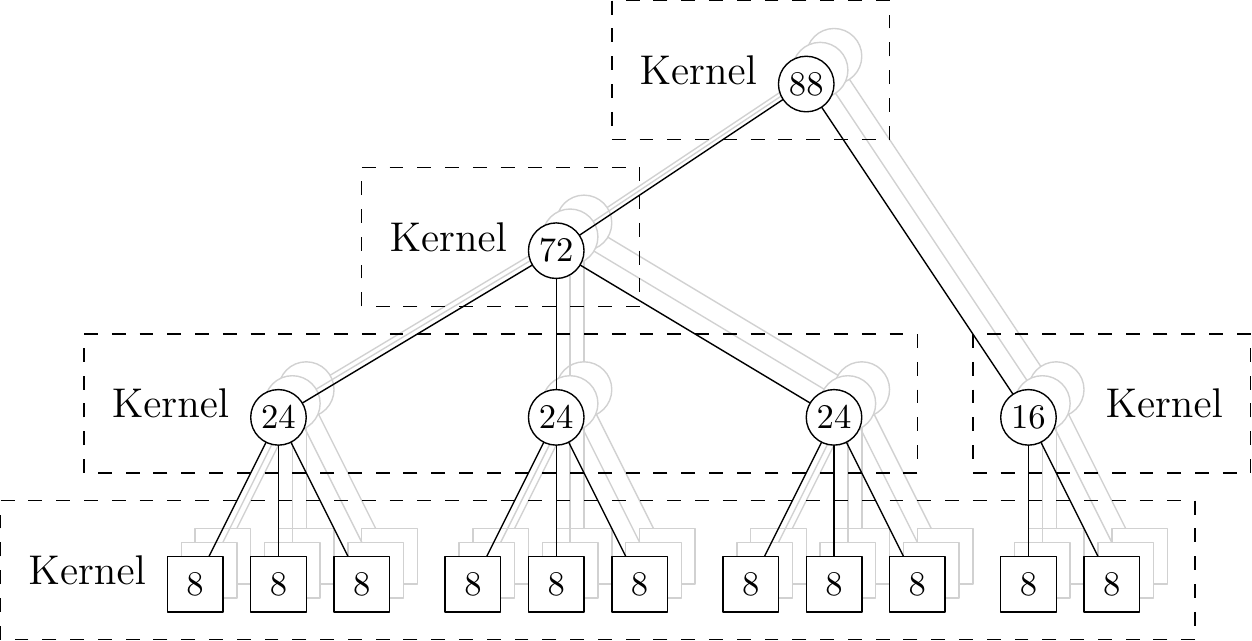}
  \caption{\label{fig:gpuTree} Illustration of how all equally-shaped splitting trees are handled together on the GPU.}
\end{figure}

\subparagraph*{Bijections.}
All leaf nodes\footnote{All leaf nodes except possibly the last of each tree, which might have fewer objects.} of all trees with the same shape are constructed with a single kernel call.
For each leaf node, we start one block of threads.
First, the threads in each block cooperate to load all objects relevant for that leaf node into the shared memory.
Similar to the SIMD implementation, where each lane tried a different hash function, now each thread tries a different hash function.
After each hash function, the threads synchronize, check if a bijection was found and if it was, store the hash function index into global memory.

\subparagraph*{Splittings.}
Like for the bijections, each splitting is handled by a thread block.
The threads cooperate to load the objects into the shared memory and then each thread tries a different hash function index.
For the two lowest aggregation levels, the thread blocks of all nodes in that level are started together using one kernel call (see \cref{fig:gpuTree}).
Note that on these levels, the size of a node and the starting seed is constant.
Therefore, the levels are very homogeneous.
Conversely, the higher levels with fanout $s=2$ are more heterogeneous.
In particular, the number of objects on a specific level may be different for different nodes on the same level.
Therefore, we launch individual kernels for each of those splittings, which contain the thread block for all trees with the same shape.
We use multiplication and shifts to increment the counters of how many objects ended up in each lane.
An alternative variant that stores counters in shared memory is slower in preliminary experiments, even when padding the counters to reduce the probability of bank conflicts.
After a valid splitting is found, the threads in a block cooperate to reorder the objects in that node accordingly.

\subparagraph*{Assembly.}
Because the kernels are launched per level, the results are stored in BFS order.
For the final data structure, we need to store them in preorder.
The CPU unpacks the resulting seeds recursively and writes them to an encoded sequence.

\section{Experiments}\label{s:experiments}
We first describe the experimental setup and general improvements.
We then continue with a comparison of different techniques of our implementation before comparing the implementation with competitors from the literature.
The code and scripts needed to reproduce our experiments are available on GitHub under the General Public License \cite{sourceCode,sourceCodeComparison}.

\subparagraph*{Experimental Setup.}
We ran most of our experiments on an Intel i7 11700 processor with 8 cores (16 hardware threads (HT)) and a base clock speed of 2.5 GHz, supporting AVX-512.
The machine runs Ubuntu 22.04 with Linux 5.15.0 and contains an NVIDIA RTX 3090 GPU.
For additional experiments, we used a machine with an AMD EPYC 7702P processor with 64 cores (128 hardware threads) and a base clock speed of 2.0 GHz.
The machine runs Ubuntu 20.04 with Linux 5.4.0 and supports only AVX2.
Unless otherwise noted, all experiments were run on the Intel machine.
We used the GNU C++ compiler version 11.2.0 with optimization flags \texttt{-O3 -march=native}.
The SIMD implementation only supports x86 CPUs and is optimized towards AVX-512 using the Vector Class Library \cite{fog2013vectorclass}.
The GPU implementation uses CUDA 11.
As a reminder, only the \emph{construction} is using SIMD, multi-threading, and/or the GPU.
The query implementation is identical for the SIMD and GPU implementation and almost equal to the original implementation \cite{esposito2020recsplit}.
We therefore did not compare the query performance of SIMD and GPU implementation.

For the comparison with competitors, we used strings of uniform random length $\in [10, 50]$ containing random characters except for the zero byte.
Note that, as a first step, all competitors generate a \emph{master hash code} (MHC) of each object using a high quality hash function.
This makes the remaining computation largely independent of the input distribution.
When only comparing different configurations of our own data structure, we used random 128-bit integers directly as MHC, which follows the approach of the original implementation~\cite{esposito2020recsplit}.

\subsection{Our Implementation}
While the original implementation \cite{esposito2020recsplit} uses std::sort to partition objects into buckets, we use IPS$^2$Ra \cite{axtmann2020engineering}.
For the less space efficient configurations ($\ell < 5, b < 100$), constructing the buckets is fast, so significant time is spent on sorting objects to buckets.
For these configurations, IPS$^2$Ra both speeds up the sequential case and also enables sorting in parallel.
For more space-efficient configurations ($\ell>8$), the partitioning step needs less than 1\% of the total construction time, both in the parallel and the sequential case.
In this section, we compare against a slight adaption of the original implementation, using IPS$^2$Ra and supporting parallel construction.

\subparagraph*{Rotation Fitting.}
In order to compare rotation fitting with the brute force variant, we give a Pareto front\footnote{A configuration is on the Pareto front if it is not dominated by any other configuration with respect to both construction time and space consumption.} of space usage versus construction time in \Cref{fig:leafMethods}.
The construction time refers to the entire MPHF construction, including the time used for splittings.
Rotation fitting is consistently faster, making the entire MPHF construction up to 3 times faster.
The space overhead of rotation fitting becomes negligible for moderately large $\ell$ (see \cref{s:proofRotations}).
Unless otherwise noted, all following experiments use rotation fitting.

\begin{figure}[t]
        \centering
    \begin{tikzpicture}
        \begin{axis}[
            xlabel={Bits per object},
            ylabel={Objects/second},
            plotLeafMethods,
            xmax=1.8,
            ymode=log,
            legend to name=paretoLeafMethodsLegend,
            legend columns=1,
          ]
          \addplot coordinates { (1.56126,838.089) (1.56408,903.942) (1.57192,2385.37) (1.57372,2389.89) (1.57713,2401.11) (1.58398,8004.09) (1.58615,8054.88) (1.59023,14169.4) (1.60541,14975.8) (1.61295,37715.9) (1.61452,37909.5) (1.61818,37957.6) (1.62544,68819.3) (1.62685,70083) (1.63137,71496) (1.64127,72319.1) (1.66,273254) (1.66157,276197) (1.66445,280332) (1.6741,286254) (1.69112,496968) (1.69325,513611) (1.69589,529269) (1.70536,547585) (1.71024,1.05552e+06) (1.71169,1.10181e+06) (1.71542,1.15821e+06) (1.7246,1.20598e+06) (1.7764,1.49298e+06) (1.77788,1.53657e+06) (1.78191,1.73491e+06) (1.79019,1.94477e+06) (1.85996,2.43191e+06) (1.86459,2.61097e+06) (1.86797,2.86862e+06) (1.8782,3.49895e+06) (1.94458,4.17014e+06) (1.94777,4.66853e+06) (2.01074,5.97372e+06) (2.08715,6.59631e+06) (2.21373,7.36377e+06) (2.2853,8.31947e+06) };
          \addlegendentry{Brute force};
          \addplot coordinates { (1.56442,1594.34) (1.57174,4312.84) (1.57403,4350.74) (1.58579,18128.9) (1.58731,18369.3) (1.59051,20675) (1.592,22406) (1.59556,22871.7) (1.60565,24707.3) (1.61414,101057) (1.61627,101438) (1.61909,103154) (1.62686,117911) (1.62857,126145) (1.63193,127763) (1.64341,133255) (1.66389,854409) (1.66695,877347) (1.66967,915081) (1.67839,965624) (1.69545,1.08131e+06) (1.69689,1.11932e+06) (1.70004,1.18315e+06) (1.7093,1.28502e+06) (1.72495,1.37552e+06) (1.72789,1.51423e+06) (1.73109,1.62496e+06) (1.74172,1.78763e+06) (1.79052,2.01776e+06) (1.79199,2.19877e+06) (1.79579,2.41663e+06) (1.80519,2.81057e+06) (1.87443,3.89712e+06) (1.93918,4.30663e+06) (2.00036,4.363e+06) (2.01172,6.18047e+06) (2.06839,6.74764e+06) (2.09022,6.86813e+06) (2.14599,7.42942e+06) (2.31651,7.49625e+06) (2.39295,8.27815e+06) };
          \addlegendentry{Rotation fitting};
        \end{axis}
    \end{tikzpicture}
    \hfill
    \begin{tikzpicture}
        \begin{axis}[
            xlabel={Bits per object},
            ylabel={Speedup},
            plotLeafMethods,
            xmax=1.8,
          ]
          \addplot coordinates { (1.56126,1.0) (1.56408,1.0) (1.57192,1.0) (1.57372,1.0) (1.57713,1.0) (1.58398,1.0) (1.58615,1.0) (1.59023,1.0) (1.60541,1.0) (1.61295,1.0) (1.61452,1.0) (1.61818,1.0) (1.62544,1.0) (1.62685,1.0) (1.63137,1.0) (1.64127,1.0) (1.66,1.0) (1.66157,1.0) (1.66445,1.0) (1.6741,1.0) (1.69112,1.0) (1.69325,1.0) (1.69589,1.0) (1.70536,1.0) (1.71024,1.0) (1.71169,1.0) (1.71542,1.0) (1.7246,1.0) (1.7764,1.0) (1.77788,1.0) (1.78191,1.0) (1.79019,1.0) (1.85996,1.0) (1.86459,1.0) (1.86797,1.0) (1.8782,1.0) (1.94458,1.0) (1.94777,1.0) (2.01074,1.0) (2.08715,1.0) (2.21373,1.0) (2.2853,1.0) };
          \addlegendentry{bruteforce};
          \addplot coordinates { (1.56442,1.71626) (1.57174,1.87607) (1.57403,1.8197) (1.58579,2.25305) (1.58731,2.00072) (1.59051,1.45768) (1.592,1.57137) (1.59556,1.58365) (1.60565,1.61816) (1.61414,2.66906) (1.61627,2.67418) (1.61909,2.56487) (1.62686,1.68237) (1.62857,1.7864) (1.63193,1.78584) (1.64341,1.68779) (1.66389,3.05672) (1.66695,3.1129) (1.66967,3.22774) (1.67839,3.0128) (1.69545,2.05342) (1.69689,2.10737) (1.70004,2.20268) (1.7093,1.43495) (1.72495,1.1391) (1.72789,1.24027) (1.73109,1.31496) (1.74172,1.38813) (1.79052,1.03655) (1.79199,1.12476) (1.79579,1.22265) (1.80519,1.38296) (1.87443,1.20399) (1.93918,1.04885) (2.00036,0.764025) (2.01172,1.03336) (2.06839,1.04912) (2.09022,1.03858) (2.14599,1.07173) };
          \addlegendentry{rotations};

          \legend{};
        \end{axis}
    \end{tikzpicture}
    \hfill
    \begin{tikzpicture}[baseline=-2cm]
        \ref*{paretoLeafMethodsLegend}
    \end{tikzpicture}
    \caption{Pareto front over the construction throughput of different variants of searching for bijections in the leaves. Single-threaded, non-vectorized measurements with $n=5$~Million objects. The plot on the right gives speedups relative to the brute force method.\footref{fn:paretoSpeedups}}
    \label{fig:leafMethods}
\end{figure}
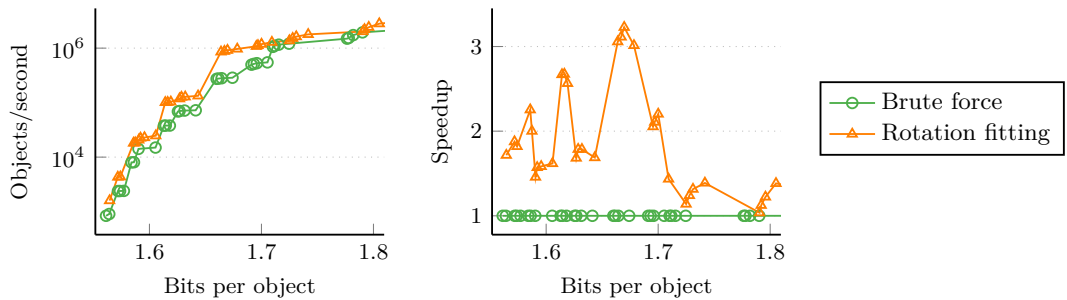

\addtocounter{footnote}{1}
\footnotetext{\label{fn:paretoSpeedups}Note that giving speedups is non-trivial here because there might not be a configuration that achieves the same space usage that we could compare with. We therefore calculate the speedup relative to an interpolation of the next larger and next smaller data points. This is reasonable since RecSplit instances can be interpolated as well by hashing a certain fraction of objects into data structures with different configurations.}

\subparagraph*{Dependence on Input Parameters.}
In \cref{fig:parameters}, we plot the throughput of the SIMD, GPU and non-vectorized versions for different leaf sizes $\ell$ and bucket sizes $b$.
For better comparability with the original paper \cite{esposito2020recsplit}, we include a wide range of configurations, even ones that are not very competitive.
The SIMD version is consistently up to 4.5 times faster than the non-vectorized version and shows the same scaling behavior in $\ell$.
The plot indicates that there is no configuration where one would prefer the non-vectorized version.
While the GPU offers significant speedups for space efficient configurations, it performs not as good for the space inefficient configurations.
A reason for this is data transfers to and from the~GPU.

\subparagraph*{Multi-Threading.}
\label{s:scalingConfigs}
\Cref{tab:showcaseOverall} shows that the parallel construction is up to 5 times faster on an 8-core machine.
In \cref{s:scalingConfigsAppendix}, we give more detailed measurements of how different RecSplit configurations scale in the number of CPU threads.
Rather unusual configurations with extremely small buckets ($b=5$) do not scale as well, because they spend a lot of time partitioning the objects to buckets -- even though we already use the highly optimized parallel sorter IPS$^2$Ra \cite{axtmann2020engineering}.

\begin{figure}[t]
        \centering
    \hspace{10mm}
    \begin{tikzpicture}[trim axis left]
        \begin{axis}[
            title={$b=5$},
            ylabel={Objects/Second},
            xlabel={leaf size $\ell$},
            plotParameters,
            ymode=log,
            legend to name=parametersLegend,
            legend columns=1,
          ]
          \addplot coordinates { (5,9.14077e+06) (6,8.77193e+06) (7,7.82473e+06) (8,6.25e+06) (9,4.59559e+06) (10,3.32889e+06) (11,2.49252e+06) (12,1.87617e+06) (13,1.46628e+06) (14,1.28041e+06) (15,1.14626e+06) (16,1.01854e+06) };
          \addlegendentry{Plain CPU};
          \addplot coordinates { (5,1.37363e+07) (6,1.57233e+07) (7,1.60256e+07) (8,1.49701e+07) (9,1.60256e+07) (10,1.57233e+07) (11,1.65017e+07) (12,1.51976e+07) (13,1.53846e+07) (14,1.65563e+07) (15,1.6129e+07) (16,1.52905e+07) };
          \addlegendentry{GPU};
          \addplot coordinates { (5,2.73224e+07) (6,2.76243e+07) (7,2.71739e+07) (8,2.51256e+07) (9,2.14592e+07) (10,1.74825e+07) (11,1.37741e+07) (12,1.08696e+07) (13,8.62069e+06) (14,7.4184e+06) (15,6.35324e+06) (16,5.63698e+06) };
          \addlegendentry{SIMD};
        \end{axis}
    \end{tikzpicture}
    \begin{tikzpicture}[trim axis left]
        \begin{axis}[
            title={$b=50$},
            xlabel={leaf size $\ell$},
            plotParameters,
            yticklabels={,,},
            ymode=log,
          ]
          \addplot coordinates { (5,7.48503e+06) (6,6.85871e+06) (7,4.34783e+06) (8,3.08642e+06) (9,2.08507e+06) (10,1.38274e+06) (11,342677) (12,287422) (13,185729) (14,117980) (15,59364.8) (16,27807.4) };
          \addlegendentry{CpuRecSplit};
          \addplot coordinates { (5,1.72414e+07) (6,1.69492e+07) (7,1.77305e+07) (8,1.74825e+07) (9,1.77936e+07) (10,1.59744e+07) (11,1.65563e+07) (12,1.65563e+07) (13,1.58228e+07) (14,1.37363e+07) (15,1.15207e+07) (16,7.59878e+06) };
          \addlegendentry{GPURecSplit};
          \addplot coordinates { (5,1.9685e+07) (6,1.92308e+07) (7,1.5528e+07) (8,1.3587e+07) (9,9.63391e+06) (10,6.64894e+06) (11,1.95618e+06) (12,1.53374e+06) (13,906290) (14,507460) (15,242166) (16,109187) };
          \addlegendentry{SIMDRecSplit};

          \legend{};
        \end{axis}
    \end{tikzpicture}
    \begin{tikzpicture}[trim axis left]
        \begin{axis}[
            title={$b=500$},
            xlabel={leaf size $\ell$},
            yticklabels={,,},
            plotParameters,
            ymode=log,
          ]
          \addplot coordinates { (5,4.55373e+06) (6,4.2735e+06) (7,2.81373e+06) (8,1.8005e+06) (9,1.21359e+06) (10,971817) (11,107177) (12,103998) (13,24681.5) (14,19368.7) (15,4379.36) (16,1647.53) };
          \addlegendentry{CpuRecSplit};
          \addplot coordinates { (5,1.62338e+07) (6,1.74825e+07) (7,1.77936e+07) (8,1.77936e+07) (9,1.71821e+07) (10,1.64474e+07) (11,1.50602e+07) (12,1.50602e+07) (13,9.19118e+06) (14,8.62069e+06) (15,3.11721e+06) (16,1.1374e+06) };
          \addlegendentry{GPURecSplit};
          \addplot coordinates { (5,1.42045e+07) (6,1.36986e+07) (7,1.11359e+07) (8,8.80282e+06) (9,6.50195e+06) (10,4.90677e+06) (11,703136) (12,595522) (13,126862) (14,92176.1) (15,20934.2) (16,8170.34) };
          \addlegendentry{SIMDRecSplit};

          \legend{};
        \end{axis}
    \end{tikzpicture}
    \begin{tikzpicture}[trim axis left]
        \begin{axis}[
            title={$b=2000$},
            xlabel={leaf size $\ell$},
            plotParameters,
            yticklabels={,,},
            ymode=log,
          ]
          \addplot coordinates { (5,2.83768e+06) (6,2.7115e+06) (7,2.02429e+06) (8,1.41844e+06) (9,1.08696e+06) (10,854701) (11,124446) (12,101752) (13,22655.3) (14,16878.2) (15,4377.39) (16,1473.28) };
          \addlegendentry{CpuRecSplit};
          \addplot coordinates { (5,1.6835e+07) (6,1.6129e+07) (7,1.71821e+07) (8,1.76678e+07) (9,1.77305e+07) (10,1.73611e+07) (11,1.50602e+07) (12,1.47493e+07) (13,9.38086e+06) (14,8.23723e+06) (15,3.01023e+06) (16,1.03563e+06) };
          \addlegendentry{GPURecSplit};
          \addplot coordinates { (5,1.0989e+07) (6,1.0846e+07) (7,8.96057e+06) (8,7.36377e+06) (9,5.60538e+06) (10,4.36681e+06) (11,648004) (12,574185) (13,115266) (14,85351.9) (15,20682.8) (16,7331.36) };
          \addlegendentry{SIMDRecSplit};

          \legend{};
        \end{axis}
    \end{tikzpicture}
    \hfill
    \begin{tikzpicture}[baseline=-2cm]
        \ref*{parametersLegend}
    \end{tikzpicture}
    \caption{Construction throughput with different hardware architectures based on different input parameters. $n=5$~Million objects, 1 CPU thread.}
    \label{fig:parameters}
\end{figure}
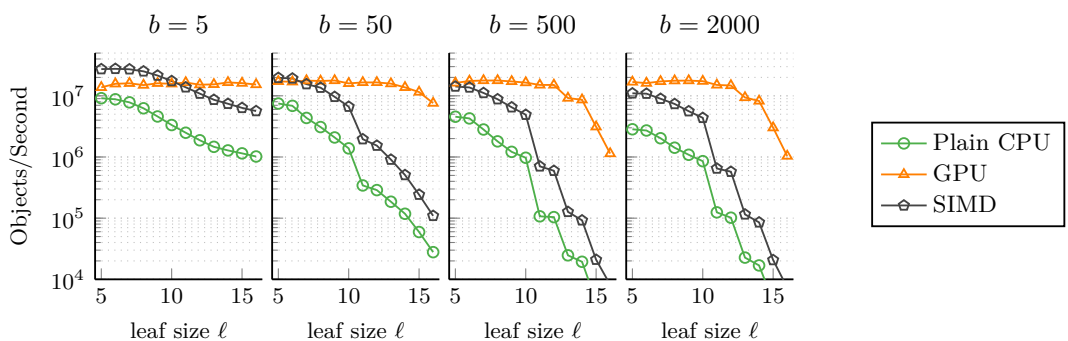

\subparagraph*{Overall Speedup.}
Our rotation fitting technique leads to a speedup of up to 3 (see \cref{fig:leafMethods}) and SIMD parallelism improves the construction speed by up to a factor of 4.5 (see \cref{fig:parameters}).
Multi-threading for highly space-efficient configurations shows a speedup of close to 5.
\cref{tab:showcaseOverall} shows the overall improvement of our implementation on CPU and GPU when compared to the original RecSplit implementation.
The original RecSplit paper says that MPHF construction at 1.56 bits per object is possible.
This configuration with $5$~Million objects takes about 1.5 hours using the original implementation.
Our SIMD implementation achieves the same space usage in just 2 minutes on the CPU and 5 seconds on the GPU.
Investing about 40 minutes of GPU time, our implementation achieves a space usage of only 1.495 bits per object.
This is about 40\% closer to the lower bound \cite{belazzougui2009hash} of $1.44$ bits, and simultaneously more than twice as fast as the original implementation.

\begin{table}[t]
  \caption{Construction time of the GPU implementation compared to our multi-threaded adaption of the original RecSplit implementation. $n=5$~Million objects (strong scaling). Construction times are given in $\mu$s/object. We do not report speedups for $\ell=24$ because the CPU baseline takes too long for this configuration.}
  \label{tab:showcaseOverall}
  \centering
  
\setlength{\tabcolsep}{4.5pt}
\begin{tabular}[t]{lllrlrr}
    \toprule
    Configuration & Method & Bijections & Threads & B/Obj & Constr. & Speedup \\ \midrule
    $\ell=16, b=2000$ & RecSplit \cite{esposito2020recsplit} &      Brute force &   1 & 1.560 & 1175.4 &    1 \\
                      &                             RecSplit &      Brute force &  16 & 1.560 &  206.5 &    5 \\
                      &                         SIMDRecSplit & Rotation fitting &   1 & 1.560 &  138.0 &    8 \\
                      &                         SIMDRecSplit & Rotation fitting &  16 & 1.560 &   27.9 &   42 \\
                      &                          GPURecSplit &      Brute force & GPU & 1.560 &    1.8 &  655 \\
                      &                          GPURecSplit & Rotation fitting & GPU & 1.560 &    1.0 & 1173 \\ \midrule
      $\ell=18, b=50$ & RecSplit \cite{esposito2020recsplit} &      Brute force &   1 & 1.707 & 2942.9 &    1 \\
                      &                             RecSplit &      Brute force &  16 & 1.713 &  504.0 &    5 \\
                      &                         SIMDRecSplit & Rotation fitting &   1 & 1.709 &   58.3 &   50 \\
                      &                         SIMDRecSplit & Rotation fitting &  16 & 1.708 &   12.3 &  239 \\
                      &                          GPURecSplit &      Brute force & GPU & 1.708 &    5.2 &  564 \\
                      &                          GPURecSplit & Rotation fitting & GPU & 1.709 &    0.5 & 5438 \\ \midrule
    $\ell=24, b=2000$ &                          GPURecSplit &      Brute force & GPU & 1.496 & 2300.9 &  --- \\
                      &                          GPURecSplit & Rotation fitting & GPU & 1.496 &  467.9 &  --- \\
    \bottomrule
\end{tabular}

\end{table}

\subparagraph*{Energy Consumption.}
\label{s:energyConsumption}
Of course, directly comparing CPU and GPU implementations is unfair.
A sensible metric to compare them is the energy consumption, which can be a major cost factor.
Additionally, the energy consumption is not influenced by market prices.
\Cref{tab:powerUsage} gives energy consumption measurements for different configurations and hardware architectures.
The energy consumption is homogeneous throughout most of the execution time, except for a short ramp-up in the beginning.
We do not count the ramp-up to the energy consumption.
Measurements are performed using a Voltcraft 870 Multimeter.

Even though SIMD instructions need slightly more power, the total energy consumption of constructing one MPHF is lower.
The GPU, even though it needs significantly more power, is so much faster that the resulting energy usage is about 1000 times lower than the original single-threaded CPU implementation.
For basic RecSplit, the AMD machine needs about 1.5 times more time than the Intel machine.
This can be readily explained by a lower clock frequency.
This performance gap grows to a factor 4.6 for sequential SIMDRecSplit.
The likely main reason is that the AMD machines lacks the AVX-512 vector units of the Intel machine.
Still, since both processors have two 256-bit AVX2 units per core, it seems that better performance might be achievable with careful tuning for the AMD architecture.
On the contrary, the AMD machine shows good scalability so that the energy consumption when using the entire machine is only a factor 1.3 larger than on the Intel machine --
despite the fact that our implementation was tuned for the Intel architecture.

\begin{table}[t]
  \caption{Energy consumption with $\ell=18, b=50$ and $n=5$~Million objects. Energy consumption is both given as difference to the idle power, as well as total energy consumption of the whole system. For CPU-only measurements of the 8-core Intel machine, we dismount the GPU.}
  \label{tab:powerUsage}
  \centering
  \begin{tabular}[t]{ll rr rr rr}
    \toprule
                 &                                      &         &           & \multicolumn{2}{c}{Total system} & \multicolumn{2}{c}{$\Delta$ to idle} \\
                                                                                \cmidrule(lr){5-6}\cmidrule(lr){7-8}
    Machine      & Method                               & Threads & Constr.   & Power &      Energy & Power &      Energy \\
                 &                                      &         & Seconds   &  Watt &       Joule &  Watt &       Joule \\ \midrule %
    8-core Intel & RecSplit \cite{esposito2020recsplit} &       1 & 14\,714.5 &    78 & 1\,147\,731 &    37 &    544\,436 \\ %
                 & SIMDRecSplit                         &       1 &     291.5 &    87 &     25\,360 &    46 &     13\,409 \\ %
                 & SIMDRecSplit                         &      16 &      61.5 &   104 &      6\,396 &    63 &      3\,874 \\ %
                 & GPURecSplit                          &         &       2.5 &   457 &      1\,142 &   380 &         950 \\ \midrule %
    64-core AMD  & RecSplit \cite{esposito2020recsplit} &       1 & 21\,620.8 &   223 & 4\,821\,438 &    91 & 1\,967\,492 \\ %
                 & SIMDRecSplit                         &       1 &  1\,328.7 &   224 &    297\,629 &    92 &    122\,240 \\ %
                 & SIMDRecSplit                         &     128 &      23.6 &   364 &      8\,590 &   232 &      5\,475 \\ %
    \bottomrule
  \end{tabular}
\end{table}

\subsection{Comparison with Competitors}
We now compare our implementation to the sequential codes RecSplit \cite{esposito2020recsplit}, SicHash \cite{lehmann2022sichash}, and CHD\cite{belazzougui2009hash} as well as the parallel codes PTHash \cite{pibiri2021pthash}, PTHash-HEM \cite{pibiri2021parallel} and BBHash \cite{limasset2017fast}.

\subparagraph*{Space usage trade-off.}
\Cref{fig:pareto} shows a space versus construction time Pareto front for each approach.
Looking at a single thread first, we make the surprising observation that SIMDRecSplit not only wins for the most space efficient configurations for which we designed it but, by far, dominates all the other methods also for less space-efficient cases.
For parallel construction, SIMDRecSplit even strengthens its margin to the competing approaches.

\begin{figure*}[t]
    \input{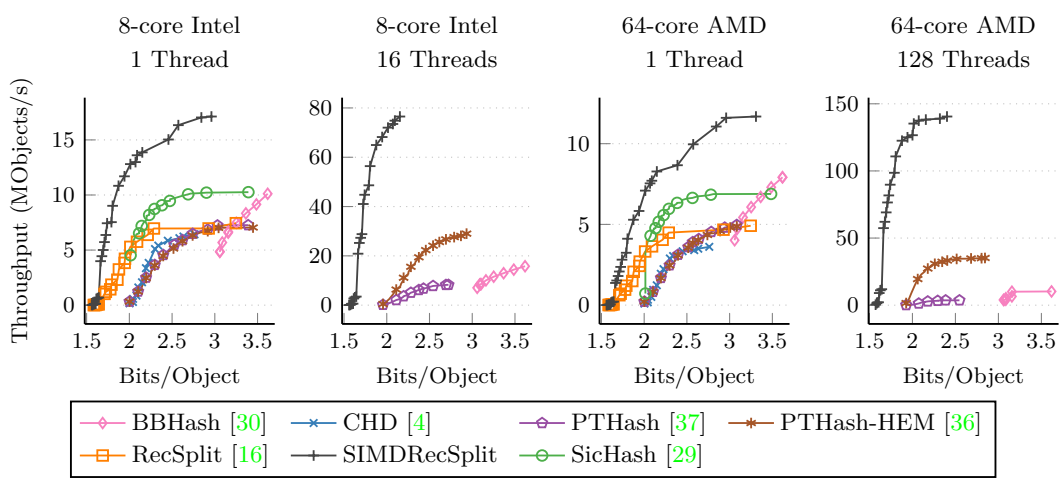}
    \caption{Trade-off of construction time vs space usage. Weak scaling, $n/p=10$ Million objects. For SicHash and PTHash, we plot all Pareto optimal data points but only show markers for every fourth point to increase readability. Therefore, the lines might bend on positions without markers.}
    \label{fig:pareto}
\end{figure*}

\begin{figure}[t]
        \centering
    \begin{tikzpicture}
        \begin{axis}[
            title={8-core Intel Machine},
            plotScaling,
            legend columns=1,
            legend to name=legendScaling,
            ylabel={Speedup},
            xlabel={Threads},
            xtick distance=4,
          ]
          \addplot[mark=diamond,color=colorBbhash,solid] coordinates { (1,1.0) (2,1.29491) (3,1.61798) (4,1.52644) (5,1.85648) (6,1.92216) (7,1.98736) (8,1.98703) (9,2.03924) (10,2.05431) (11,2.02782) (12,2.0714) (13,2.07608) (14,2.08586) (15,2.09868) (16,2.11166) };
          \addlegendentry{BBHash \cite{limasset2017fast}};
          \addplot[mark=pentagon,color=colorPthash,solid] coordinates { (1,1.0) (2,1.4131) (3,1.70486) (4,1.91447) (5,2.04784) (6,2.18374) (7,2.25253) (8,2.25484) (9,2.18663) (10,2.18356) (11,2.17942) (12,2.22413) (13,2.22919) (14,2.25772) (15,2.29812) (16,2.29653) };
          \addlegendentry{PTHash \cite{pibiri2021pthash}};
          \addplot[mark=asterisk,color=colorPthashHem,solid] coordinates { (1,1.0) (2,1.90629) (3,2.63499) (4,3.3488) (5,3.93104) (6,4.52187) (7,4.97211) (8,5.36362) (9,5.21925) (10,5.58109) (11,5.87354) (12,6.12967) (13,6.25524) (14,6.42311) (15,6.57412) (16,6.82404) };
          \addlegendentry{PTHash-HEM \cite{pibiri2021parallel}};
          \addplot[mark=+,color=colorSimdRecSplit,solid] coordinates { (1,1.0) (2,1.96404) (3,2.83236) (4,3.23216) (5,4.24781) (6,4.71789) (7,5.31983) (8,5.90526) (9,5.09263) (10,5.37661) (11,5.45865) (12,5.89539) (13,5.7808) (14,5.66927) (15,6.11899) (16,6.48385) };
          \addlegendentry{SIMDRecSplit};

          \addplot[color=gray,dashed] coordinates { (8,1) (8,6.8) };
          \node[color=gray] at (axis cs: 9,6.5) {\tiny HT};
        \end{axis}
    \end{tikzpicture}
    \hfill
    \begin{tikzpicture}
        \begin{axis}[
            title={64-core AMD Machine},
            xlabel={Threads},
            plotScaling,
            xtick distance=32,
          ]
          \addplot[mark=diamond,color=colorBbhash,solid] coordinates { (1,1.0) (8,2.1758) (16,2.63089) (24,2.88812) (32,2.94426) (40,2.24901) (48,2.98349) (56,2.98996) (64,2.96929) (72,2.95656) (80,2.92487) (88,2.91381) (96,2.87481) (104,2.84639) (112,2.84484) (120,2.83407) (128,2.77347) };
          \addlegendentry{BBHash \cite{limasset2017fast}};
          \addplot[mark=pentagon,color=colorPthash,solid] coordinates { (1,1.0) (8,1.16756) (16,1.20284) (24,1.28861) (32,1.29787) (40,1.30161) (48,1.29184) (56,1.29321) (64,1.2778) (72,1.29598) (80,1.29655) (88,1.28168) (96,1.25078) (104,1.22317) (112,1.18453) (120,1.18374) (128,1.13289) };
          \addlegendentry{PTHash \cite{pibiri2021pthash}};
          \addplot[mark=asterisk,color=colorPthashHem,solid] coordinates { (1,1.0) (8,5.97954) (16,9.1241) (24,11.4663) (32,12.4881) (40,12.7593) (48,13.2779) (56,13.6647) (64,13.4033) (72,14.0646) (80,14.2612) (88,14.3598) (96,14.3953) (104,14.4994) (112,14.7382) (120,14.7794) (128,14.7607) };
          \addlegendentry{PTHash-HEM \cite{pibiri2021parallel}};
          \addplot[mark=+,color=colorSimdRecSplit,solid] coordinates { (1,1.0) (8,7.55383) (16,14.2636) (24,20.8273) (32,26.443) (40,30.0276) (48,31.6439) (56,31.1455) (64,32.4463) (72,34.2767) (80,34.6259) (88,35.7058) (96,38.1213) (104,38.6496) (112,39.6211) (120,39.8388) (128,39.463) };
          \addlegendentry{SIMDRecSplit};

          \addplot[color=gray,dashed] coordinates { (64,1) (64,45) };
          \node[color=gray] at (axis cs: 72,43) {\tiny HT};
          \legend{};
        \end{axis}
    \end{tikzpicture}
    \hfill
        \begin{tikzpicture}[baseline=-2cm]
            \ref*{legendScaling}
        \end{tikzpicture}
    \caption{Construction time speedups when using multiple threads $t$. Strong scaling, $n=50$~Million. Speedups are given relative to each method's single threaded performance. For a comparison of absolute performance, refer to \cref{fig:pareto}.
    Configurations used are BBHash: $\gamma=2.0$; PTHash/PTHash-HEM: $c=6.0$, $\alpha=0.95$, EF; SIMDRecSplit: $\ell=10$, $b=2000$.}
    \label{fig:scalingCompetitors}
\end{figure}
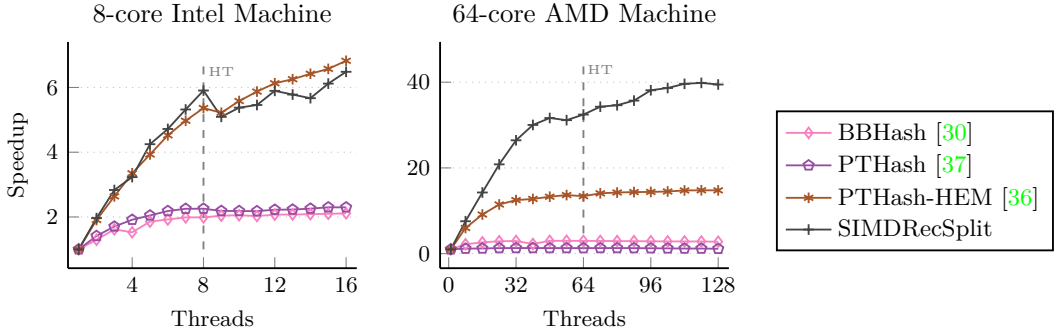

\subparagraph*{Construction Scaling.}
\Cref{fig:scalingCompetitors} compares scaling behavior of the parallel codes.
We see that BBHash scales poorly while both PTHash-HEM and SIMDRecSplit scale well on the 8-core Intel machine.
However, SIMDRecSplit scales better than PTHash-HEM on the 64-core AMD machine.

\begin{table}[t]
  \caption{Query and construction time of different competitor configurations on 10~Million objects.}
  \label{tab:queries}
  \centering
  
\begin{tabular}[t]{ll rr}
    \toprule
    Method & Bits/Obj. & Constr./Obj. & Query/Obj. \\ \midrule

                                BBHash \cite{limasset2017fast}, $\gamma$=$5.0$ & 6.871 &       50 ns &  36 ns \\
                                BBHash \cite{limasset2017fast}, $\gamma$=$1.0$ & 3.059 &      208 ns &  51 ns \\\midrule
              PTHash \cite{pibiri2021pthash}, $c$=$11.0$, $\alpha$=$0.88$, D-D & 4.379 &      138 ns &  25 ns \\
               PTHash \cite{pibiri2021pthash}, $c$=$7.0$, $\alpha$=$0.99$, C-C & 3.313 &      199 ns &  20 ns \\
                PTHash \cite{pibiri2021pthash}, $c$=$6.0$, $\alpha$=$0.99$, EF & 2.345 &      248 ns &  35 ns \\\midrule
     SicHash \cite{lehmann2022sichash}, $\alpha$=$0.9$, $p_1$=$20$, $p_2$=$77$ & 2.412 &      119 ns &  41 ns \\
    SicHash \cite{lehmann2022sichash}, $\alpha$=$0.97$, $p_1$=$44$, $p_2$=$30$ & 2.081 &      172 ns &  40 ns \\\midrule
                     RecSplit \cite{esposito2020recsplit}, $\ell$=$5$, $b$=$5$ & 2.928 &      145 ns &  65 ns \\
                   RecSplit \cite{esposito2020recsplit}, $\ell$=$8$, $b$=$100$ & 1.793 &      709 ns &  75 ns \\
                 RecSplit \cite{esposito2020recsplit}, $\ell$=$14$, $b$=$2000$ & 1.584 & 126\,534 ns &  96 ns \\\midrule
                                             SIMDRecSplit, $\ell$=$5$, $b$=$5$ &  2.96 &       49 ns &  71 ns \\
                                           SIMDRecSplit, $\ell$=$8$, $b$=$100$ & 1.806 &      107 ns &  80 ns \\
                                         SIMDRecSplit, $\ell$=$14$, $b$=$2000$ & 1.585 &  11\,742 ns & 110 ns \\

    \bottomrule
\end{tabular}

\end{table}

\subparagraph*{Queries.}
\Cref{tab:queries} shows that when looking at the query time, PTHash is a clear winner.
While BBHash can achieve the same query speed and good construction speed, its space usage is large.
SicHash has a query time close to PTHash's most compact representation, but is faster to construct and more space efficient.
All RecSplit variants can achieve significantly lower space than other competitors but require considerably more query time.
Our single-threaded SIMD implementation dominates most competitors with respect to both space and construction time.
The use of rotations makes the queries about 10\% slower than the original RecSplit implementation.
The main goal of RecSplit is to achieve extremely small representation, and queries are not very fast to begin with, so this seems acceptable.

\section{Conclusion and Future Work}\label{s:conclusion}

We have shown that by harnessing parallelism at
all available levels -- bits, vectors, cores, and
GPUs -- one can dramatically accelerate the
construction of highly space efficient minimal
perfect hash functions (MPHFs) using the
brute force RecSplit approach \cite{esposito2020recsplit}.
This leads to speedups of up to \SimdSpeedup{} on SIMD and \GpuSpeedup{} on the GPU and
also dramatically reduces energy consumption.
Surprisingly, this
even turns out to be the fastest available
approach for constructing less space-efficient
MPHFs. This is not what we expected. Our initial
hypothesis was that there would be a trade-off with
asymptotically faster approaches winning for fewer
requirements on space consumption.
Our new technique \emph{rotation fitting}
reduces the work needed per
tried hash function while adding a tiny bit
of space requirement. The asymptotically
``obvious'' improvement of replacing $\ell$
rotations/checks by two table lookups are not
productive on current architectures.  So,
brute force, simplicity (in the inner loops), and
parallelism currently wins against any attempt at
algorithmic sophistication.

Another attempt at sophistication that so far
failed is to combine brute force RecSplit with the
retrieval approach of SicHash
\cite{lehmann2022sichash}. The idea of this
\emph{ShockHash} approach is to allow retrieval of
a single bit of information for each element.  The
brute force part then tries pairs of random hash
functions until they define a pseudo-forest -- a
collection of components consisting of a tree plus
one additional edge.  While ShockHash seems to
allow space efficient perfect hashing, initial
experiments indicated that performance-wise
ShockHash is also inferior to pure brute force
(see \cref{s:shockHash} for details).  More
efficient implementations of ShockHash may change
this picture in the future.

Also, rotation fitting could be
generalized by splitting into more than two
parts. The resulting search for several rotations
gives more room for sophistications like search
space pruning. Furthermore, the approach from rotation fitting to use
a lookup table for normalizing bit patterns could
be generalized to a richer set of mappings than
just rotations.  

Everything discussed so far is mainly concerned with construction time.
However, an equally important 
problem is to improve query time.  Traversing an
aggressively compressed tree for each
query is inherently more expensive
than the simple constant time operations needed in
PTHash \cite{pibiri2021pthash} or SicHash
\cite{lehmann2022sichash} but there should be more
efficient ways to break down MPHF construction
into small subproblems that can be solved with
brute force. We believe that the techniques
developed here will turn out to be useful in that
respect.

Finally, we can look for generalizations of
RecSplit for computing non-minimal PHFs which
allows us to further reduce space consumption of
the hash function itself.
Better tuning of the SIMD variant for AMD
  or perhaps even a portable implementation that
  also works on ARM or RISC-V would be relevant for
  widespread application.

\bibliography{paper}

\clearpage
\appendix
\section{Probability of Finding a Bijection}\label{s:proofRotations}
In this section, we show that rotation fitting improves the construction time by a factor close to $m$, while having negligible space overhead.
Refer to \cref{fig:prob} for numeric evaluations of the formula in the proof below.

\begin{figure}[t]
  \centering
      \centering
    \begin{tikzpicture}
        \begin{axis}[
            xlabel={$m$},
            ylabel={\begin{tabular}{c}
                      Expected factor\\ higher probability
                    \end{tabular}},
            plotProbabilities,
          ]
          \addplot coordinates { (2,1.5) (3,2.5) (4,3.25) (5,4.75) (6,5.0625) (7,6.90625) (8,7.453125) (9,8.771875) (10,9.63238) (11,10.9902) (12,11.6976) (13,12.9971) (14,13.8828) (15,14.961) (16,15.9297) (17,16.9998) (18,17.955) (19,18.9999) (20,19.9782) (21,20.9971) (22,21.9892) (23,23) (24,23.9932) (25,24.9999) (26,25.9968) };
          \addlegendentry{data=expected\_factor\_probability};

          \legend{};
        \end{axis}
    \end{tikzpicture}
    \hspace{1cm}
    \begin{tikzpicture}
        \begin{axis}[
            xlabel={$m$},
            ylabel={\begin{tabular}{c}Expec. space overhead\\ (Bits/Object)\end{tabular}},
            plotProbabilities,
          ]
          \addplot coordinates { (2,0.292481) (3,0.194988) (4,0.146241) (5,0.0643856) (6,0.0653862) (7,0.018469) (8,0.0212406) (9,0.00853507) (10,0.00768156) (11,0.00127459) (12,0.00369951) (13,0.000324651) (14,0.00102637) (15,0.000328179) (16,0.000439204) (17,2.07163e-05) (18,0.000211298) (19,5.21361e-06) (20,8.1405e-05) (21,1.0744e-05) (22,3.29499e-05) (23,3.29009e-07) (24,1.73104e-05) (25,2.8893e-07) };
          \addlegendentry{data=space\_overhead\_rotational\_fitting};

          \legend{};
        \end{axis}
    \end{tikzpicture}
  \caption{Expected factor of higher probability to find a bijection using rotation fitting (left) and expected space overhead introduced by storing hash function index \emph{and} rotation compared to RecSplit's brute force approach.}
  \label{fig:prob}
\end{figure}
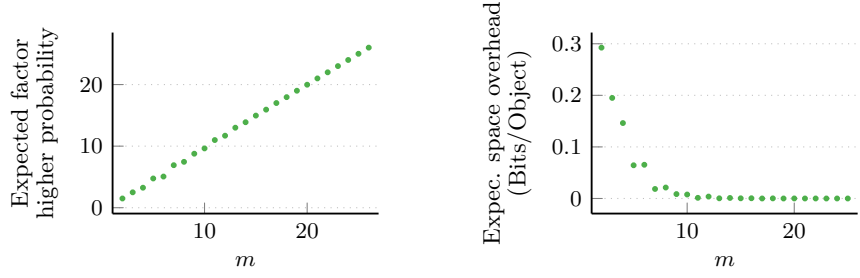

  \newcommand{\bba}{\ensuremath{\mathbb{a}}}
  \newcommand{\bbb}{\ensuremath{\mathbb{b}}}
  \begin{lemma}
    Let \(|A|=\bba\), \(|B|=\bbb\), and \(\mathds{P}(R)\) be the probability of finding a bijection using rotation fitting.
    Furthermore, let \(\mathds{P}(B)\) denote the probability of finding a bijection using RecSplit's brute force strategy.
    If \(\bba\) and \(\bbb\) are relatively prime, then \(\mathds{P}(R)\geq m\mathds{P}(B)\).
    Otherwise, \(\mathds{P}(R)\to m\mathds{P}(B)\) for \(m\to\infty\).
  \end{lemma}
  \begin{proof}
    First, we consider the number of different injective functions under cyclic shifts, i.e., equivalence classes under rotation.
    We have a bit vector of length \(m\) with \bbb\ set bits (and \bba\ unset bits).
    Then, the total number of equivalence classes under rotation is \(\frac{1}{m}\sum_{d \textnormal{ divides gcd}(\bba,\bbb)}\phi(d)\binom{m/d}{\bbb/d}\),
    where $\textnormal{gcd}$ gives the greatest common divisor.
    The probability of the event \(\mathcal{I}\) that there is an \(r\) such that \(a | \mathrm{rot}^r_m(b)\) has the \(m\) least significant bits set is
    \[\mathds{P}(\mathcal{I})\geq m\frac{1}{\sum_{d\textnormal{ divides gcd}(\bba,\bbb)}\phi(d)\binom{m/d}{\bbb/d}},\]
    where \(\phi(i)=|\{j\leq i\colon \textnormal{gcd}(i,j)=1\}|\) is Euler's totient function.
    Now, we determine the probability \(\mathds{P}(R)\) using the events
    \(\mathcal{A}\): \texttt{popcount}(a)=\bba\ and \(\mathcal{B}\): \texttt{popcount}(b)=\bbb.
    \begin{align*}
      \mathds{P}(R)&=    \mathds{P}(\mathcal{A})\mathds{P}(\mathcal{B})\mathds{P}(\mathcal{I})\\
                   &\geq \frac{m!}{(m-\bba)!m^{\bba}}\cdot\frac{m!}{(m-\bbb)!m^{\bbb}}\cdot\mathds{P}(\mathcal{I})
                    =    \frac{m!}{m^m}\cdot\frac{m!}{\bba!\bbb!}\cdot\mathds{P}(\mathcal{I})=\mathds{P}(B)\cdot\frac{m!}{\bba!\bbb!}\cdot\mathds{P}(\mathcal{I})\\
                   &\geq \mathds{P}(B)\cdot\frac{m!}{\bba!\bbb!}\cdot m\frac{1}{\sum_{d\textnormal{ divides gcd}(\bba,\bbb)}\phi(d)\binom{m/d}{b/d}}\\
                   &=    \mathds{P}(B)\cdot m\cdot \frac{m!}{m!+(\bba!\bbb!)\sum_{d\textnormal{ divides gcd}(\bba,\bbb), d\neq 1}\phi(d)\binom{m/d}{b/d}}\\
                   &=    \mathds{P}(B)\cdot m\cdot \frac{1}{1+\sum_{d\textnormal{ divides gcd}(\bba,\bbb), d\neq 1}\phi(d)\frac{(m/d)!\bba!\bbb!}{m!(\bba/d)!(\bbb/d)!}}\\
                   &\sim \mathds{P}(B)\cdot m\cdot \frac{1}{1+\sum_{d\textnormal{ divides gcd}(\bba,\bbb), d\neq 1}\phi(d)\sqrt{d}\frac{\bba^{\bba-\bba/d}\bbb^{\bbb-\bbb/d}}{m^{m-m/d}}}\\
                   &\to  \mathds{P}(B)\cdot m\textnormal{~for~}m\to\infty
    \end{align*}
    Note that if \bba\ and \bbb\ are relatively prime the sum is zero, as {\(\textnormal{gcd}(\bba,\bbb)=1\)}.
  \end{proof}

\section{Scaling} \label{s:scalingConfigsAppendix}
\Cref{fig:scalingConfigs} shows how the SIMD version scales when selecting a different number of CPU threads.
The configurations are adopted from the RecSplit paper \cite{esposito2020recsplit}.
On the Intel machine, we can see that the most space efficient configuration scales best, closely followed by the other configurations.
Only a variant with extremely small buckets ($b=5$) does not scale as well.
In this case, the entire construction is dominated by partitioning the objects to buckets.
Given that we already use the highly optimized sorter IPS$^2$Ra \cite{axtmann2020engineering} and that this is a rather unusual RecSplit configuration with a lot of space overhead, having non-optimal speedups here is acceptable.
On the AMD machine, the difference between the different configurations is slightly more pronounced.

\begin{figure}[t]
    \centering

    \begin{tikzpicture}
        \begin{axis}[
            title={8-core Intel Machine},
            xlabel={Threads},
            ylabel={Speedup},
            plotScalingConfigs,
            legend columns=1,
            legend to name=legendScalingConfigs,
            ytick distance=1,
          ]
          \addplot coordinates { (1,1.0) (2,1.73143) (4,2.43545) (6,2.75771) (8,2.95371) (10,3.06254) (12,3.13579) (14,3.16577) (16,3.14143) };
          \addlegendentry{$\ell=5$, $b=5$ ($n=2$$\cdot$$10^9$)};
          \addplot coordinates { (1,1.0) (2,1.9106) (4,2.83032) (6,3.50032) (8,3.95445) (10,3.85432) (12,4.06196) (14,4.15802) (16,4.28118) };
          \addlegendentry{$\ell=8$, $b=100$ ($n=5$$\cdot$$10^8$)};
          \addplot coordinates { (1,1.0) (2,1.80295) (4,2.8497) (6,3.61861) (8,4.23481) (10,3.9172) (12,4.14479) (14,4.28171) (16,4.38441) };
          \addlegendentry{$\ell=12$, $b=9$ ($n=2$$\cdot$$10^8$)};
          \addplot coordinates { (1,1.0) (2,1.84964) (4,2.99015) (6,3.91628) (8,4.72807) (10,4.22609) (12,4.66785) (14,4.79084) (16,5.06277) };
          \addlegendentry{$\ell=16$, $b=2000$ ($n=5$$\cdot$$10^5$)};

          \addplot[color=gray,dashed] coordinates { (8,1) (8,5.2) };
          \node[color=gray] at (axis cs: 9,5.0) {\tiny HT};
        \end{axis}
    \end{tikzpicture}
    \hfill
    \begin{tikzpicture}
        \begin{axis}[
            title={64-core AMD Machine},
            xlabel={Threads},
            plotScalingConfigs,
          ]
          \addplot coordinates { (1,1.0) (16,5.99939) (32,7.06479) (48,7.48198) (64,7.52038) (80,7.66004) (96,7.6989) (112,7.71334) (128,7.64579) };
          \addlegendentry{$\ell$=$5$, $b$=$5$};
          \addplot coordinates { (1,1.0) (16,12.9933) (32,21.1469) (48,25.5027) (64,25.9867) (80,27.917) (96,28.3753) (112,29.2206) (128,30.1967) };
          \addlegendentry{$\ell$=$8$, $b$=$100$};
          \addplot coordinates { (1,1.0) (16,14.7744) (32,26.9098) (48,34.8715) (64,37.3352) (80,38.0622) (96,40.4099) (112,42.7599) (128,45.0598) };
          \addlegendentry{$\ell$=$12$, $b$=$9$};
          \addplot coordinates { (1,1.0) (16,13.6386) (32,26.5655) (48,33.149) (64,41.3189) (80,41.2611) (96,41.6834) (112,42.6082) (128,48.0024) };
          \addlegendentry{$\ell$=$16$, $b$=$2000$};

          \addplot[color=gray,dashed] coordinates { (64,1) (64,54) };
          \node[color=gray] at (axis cs: 72,49) {\tiny HT};

          \legend{};
        \end{axis}
    \end{tikzpicture}
    \hfill
    \begin{tikzpicture}[baseline=-2.3cm]
        \ref*{legendScalingConfigs}
    \end{tikzpicture}
    \caption{Construction speedup by number of threads used, for different configurations. The number of input objects $n$ is selected such that construction takes a similar amount of time on all configurations. Configurations are the examples that are highlighted in the RecSplit paper \cite{esposito2020recsplit}.}
    \label{fig:scalingConfigs}
\end{figure}
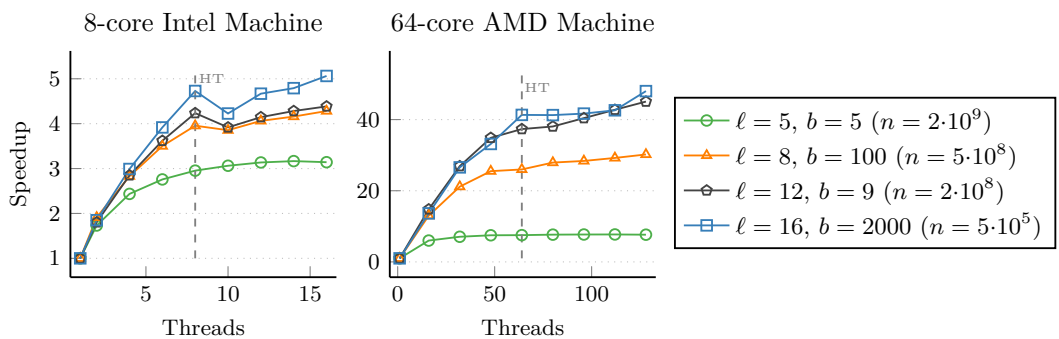

\section{ShockHash} \label{s:shockHash}
ShockHash (small, heavily overloaded cuckoo hash tables) provides an additional method for finding bijections and is based on an idea introduced in SicHash \cite{lehmann2022sichash}.
SicHash already \emph{overloads} cuckoo hash tables beyond their asymptotic maximal load factor, but here we drive this idea to its extremes.
In a (binary) cuckoo hash table \cite{pagh2004cuckoo}, each object can be placed at two different positions, determined by two hash functions.
When we create a cuckoo hash table of size $m$ and then successfully insert $m$ objects, the object positions implicitly describe a bijection.
We can then use a 1-bit retrieval data structure that maps each object to a bit indicating which of the hash functions was used to place it.
The retrieval data structure can be stored with space close to 1 bit per object \cite{DHSW22}.
Since, asymptotically, such cuckoo hash tables can support only $m/2$ objects, the success probability $p$ tends to 0 as $m$ gets large.
However, $p$ is much larger than the probability for directly finding a bijection using brute force as in basic RecSplit.
Indeed, our experiments indicate that $p\rightarrow 2^{-0.44m}$ for large $m$.
Overall, the space for the retrieval data structure plus the space for encoding the choice of hash function seem to converge to the lower bound of 1.44 bits per object.
ShockHash reduces the need for undirected, brute force searching significantly and replaces it with the more directed construction of cuckoo hash tables.

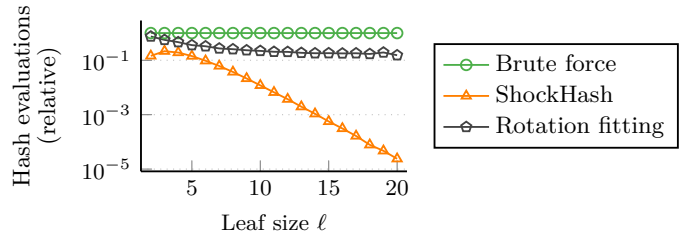
\begin{figure}[t]
        \centering
    \begin{tikzpicture}
        \begin{axis}[
            xlabel={Leaf size $\ell$},
            ylabel={\begin{tabular}{c}Hash evaluations\\ (relative)\end{tabular}},
            plotHfEvals,
            ymode=log,
            legend columns=1,
            legend style={at={(1.1, 0.5)},anchor=west},
          ]
          \addplot coordinates { (2,1.0) (3,1.0) (4,1.0) (5,1.0) (6,1.0) (7,1.0) (8,1.0) (9,1.0) (10,1.0) (11,1.0) (12,1.0) (13,1.0) (14,1.0) (15,1.0) (16,1.0) (17,1.0) (18,1.0) (19,1.0) (20,1.0) };
          \addlegendentry{Brute force};
          \addplot coordinates { (2,0.14507) (3,0.215694) (4,0.192306) (5,0.142792) (6,0.0958738) (7,0.0617299) (8,0.0371929) (9,0.0214151) (10,0.0119559) (11,0.00664028) (12,0.00372006) (13,0.00196426) (14,0.00108054) (15,0.000563638) (16,0.000310436) (17,0.00016479) (18,7.81211e-05) (19,4.69162e-05) (20,2.39372e-05) };
          \addlegendentry{ShockHash};
          \addplot coordinates { (2,0.749847) (3,0.56226) (4,0.464262) (5,0.360423) (6,0.327374) (7,0.27244) (8,0.25541) (9,0.23375) (10,0.221337) (11,0.205367) (12,0.203541) (13,0.187957) (14,0.181493) (15,0.182725) (16,0.17885) (17,0.181346) (18,0.170712) (19,0.193733) (20,0.153462) };
          \addlegendentry{Rotation fitting};
        \end{axis}
    \end{tikzpicture}
    \caption{Average number of hash function evaluations needed to find a leaf bijection, relative to the brute force method.}
    \label{fig:hashFunctionEvals}
\end{figure}

\Cref{fig:hashFunctionEvals} compares the average number of hash function evaluations needed to find a leaf bijection between ordinary brute force, rotation fitting, and a prototypical implementation of the ShockHash approach.
It shows that rotation fitting is able to reduce the number of hash function evaluations by a linear factor.
At the same time, ShockHash requires exponentially fewer hash function evaluations than the brute force method.

\begin{figure}[t]
        \centering
    \begin{tikzpicture}
        \begin{axis}[
            xlabel={Bits per object},
            ylabel={Objects/second},
            plotLeafMethods,
            xmax=1.8,
            ymode=log,
            legend to name=paretoLeafMethodsShockHashLegend,
            legend columns=1,
          ]
          \addplot coordinates { (1.56126,838.089) (1.56408,903.942) (1.57192,2385.37) (1.57372,2389.89) (1.57713,2401.11) (1.58398,8004.09) (1.58615,8054.88) (1.59023,14169.4) (1.60541,14975.8) (1.61295,37715.9) (1.61452,37909.5) (1.61818,37957.6) (1.62544,68819.3) (1.62685,70083) (1.63137,71496) (1.64127,72319.1) (1.66,273254) (1.66157,276197) (1.66445,280332) (1.6741,286254) (1.69112,496968) (1.69325,513611) (1.69589,529269) (1.70536,547585) (1.71024,1.05552e+06) (1.71169,1.10181e+06) (1.71542,1.15821e+06) (1.7246,1.20598e+06) (1.7764,1.49298e+06) (1.77788,1.53657e+06) (1.78191,1.73491e+06) (1.79019,1.94477e+06) (1.85996,2.43191e+06) (1.86459,2.61097e+06) (1.86797,2.86862e+06) (1.8782,3.49895e+06) (1.94458,4.17014e+06) (1.94777,4.66853e+06) (2.01074,5.97372e+06) (2.08715,6.59631e+06) (2.21373,7.36377e+06) (2.2853,8.31947e+06) };
          \addlegendentry{Brute force};
          \addplot coordinates { (1.58496,4436.68) (1.58982,4482.34) (1.59221,7961.73) (1.59645,8043.51) (1.60013,14364) (1.60456,14550) (1.60994,45167.9) (1.6143,45627.1) (1.61953,77465.3) (1.62496,78644.8) (1.63178,127288) (1.63593,128942) (1.64419,129628) (1.64801,198083) (1.65187,204901) (1.65945,294950) (1.66123,299958) (1.6646,302444) (1.66485,305904) (1.66531,328926) (1.66717,333667) (1.66997,335166) (1.67511,369604) (1.67698,427533) (1.68566,485437) (1.68703,493535) (1.69033,500601) (1.69538,558160) (1.69714,569022) (1.70037,585823) (1.71108,623752) (1.71299,640533) (1.71677,661989) (1.72098,765814) (1.72265,803342) (1.72548,838785) (1.73589,890155) (1.75997,893017) (1.76806,999800) (1.7782,1.01461e+06) (1.78684,1.05932e+06) (1.79951,1.06564e+06) (1.80341,1.08719e+06) (1.80368,1.19161e+06) (1.80456,1.20395e+06) (1.82606,1.21036e+06) (1.8322,1.28966e+06) (1.83513,1.30993e+06) (1.84704,1.31926e+06) (1.86483,1.4339e+06) (1.86971,1.45943e+06) (1.90036,1.57729e+06) (1.93648,1.58328e+06) (1.93911,1.72652e+06) (1.97006,1.85667e+06) (1.99882,1.88893e+06) (2.01317,2.35294e+06) (2.04874,2.36072e+06) (2.08971,2.87356e+06) (2.29107,3.003e+06) (2.32991,3.37838e+06) (2.40994,3.61795e+06) (2.60762,3.7037e+06) };
          \addlegendentry{ShockHash};
          \addplot coordinates { (1.56442,1594.34) (1.57174,4312.84) (1.57403,4350.74) (1.58579,18128.9) (1.58731,18369.3) (1.59051,20675) (1.592,22406) (1.59556,22871.7) (1.60565,24707.3) (1.61414,101057) (1.61627,101438) (1.61909,103154) (1.62686,117911) (1.62857,126145) (1.63193,127763) (1.64341,133255) (1.66389,854409) (1.66695,877347) (1.66967,915081) (1.67839,965624) (1.69545,1.08131e+06) (1.69689,1.11932e+06) (1.70004,1.18315e+06) (1.7093,1.28502e+06) (1.72495,1.37552e+06) (1.72789,1.51423e+06) (1.73109,1.62496e+06) (1.74172,1.78763e+06) (1.79052,2.01776e+06) (1.79199,2.19877e+06) (1.79579,2.41663e+06) (1.80519,2.81057e+06) (1.87443,3.89712e+06) (1.93918,4.30663e+06) (2.00036,4.363e+06) (2.01172,6.18047e+06) (2.06839,6.74764e+06) (2.09022,6.86813e+06) (2.14599,7.42942e+06) (2.31651,7.49625e+06) (2.39295,8.27815e+06) };
          \addlegendentry{Rotation fitting};
        \end{axis}
    \end{tikzpicture}
    \hfill
    \begin{tikzpicture}
        \begin{axis}[
            xlabel={Bits per object},
            ylabel={Speedup},
            plotLeafMethods,
            xmax=1.8,
          ]
          \addplot coordinates { (1.56126,1.0) (1.56408,1.0) (1.57192,1.0) (1.57372,1.0) (1.57713,1.0) (1.58398,1.0) (1.58615,1.0) (1.59023,1.0) (1.60541,1.0) (1.61295,1.0) (1.61452,1.0) (1.61818,1.0) (1.62544,1.0) (1.62685,1.0) (1.63137,1.0) (1.64127,1.0) (1.66,1.0) (1.66157,1.0) (1.66445,1.0) (1.6741,1.0) (1.69112,1.0) (1.69325,1.0) (1.69589,1.0) (1.70536,1.0) (1.71024,1.0) (1.71169,1.0) (1.71542,1.0) (1.7246,1.0) (1.7764,1.0) (1.77788,1.0) (1.78191,1.0) (1.79019,1.0) (1.85996,1.0) (1.86459,1.0) (1.86797,1.0) (1.8782,1.0) (1.94458,1.0) (1.94777,1.0) (2.01074,1.0) (2.08715,1.0) (2.21373,1.0) (2.2853,1.0) };
          \addlegendentry{bruteforce};
          \addplot coordinates { (1.58496,0.552724) (1.58982,0.340472) (1.59221,0.557951) (1.59645,0.555144) (1.60013,0.978138) (1.60456,0.974668) (1.60994,1.92353) (1.6143,1.20444) (1.61953,1.87066) (1.62496,1.2042) (1.63178,1.77951) (1.63593,1.79404) (1.64419,1.58696) (1.64801,2.01423) (1.65187,1.6542) (1.65945,1.16747) (1.66123,1.08856) (1.6646,1.07853) (1.66485,1.09028) (1.66531,1.17118) (1.66717,1.18332) (1.66997,1.18145) (1.67511,1.25869) (1.67698,1.38639) (1.68566,1.20746) (1.68703,1.16876) (1.69033,1.04173) (1.69538,1.0608) (1.69714,1.07036) (1.70037,1.08934) (1.71108,0.576562) (1.71299,0.57148) (1.71677,0.568232) (1.72098,0.645341) (1.72265,0.671967) (1.72548,0.693249) (1.73589,0.707191) (1.75997,0.643292) (1.76806,0.695325) (1.7782,0.654314) (1.78684,0.571362) (1.79951,0.533292) (1.80341,0.537816) (1.80368,0.588995) (1.80456,0.59353) (1.82606,0.558274) (1.8322,0.583159) (1.83513,0.58666) (1.84704,0.567643) (1.86483,0.54568) (1.86971,0.493168) (1.90036,0.426567) (1.93648,0.388558) (1.93911,0.420564) (1.97006,0.36694) (1.99882,0.332941) (2.01317,0.3927) (2.04874,0.376635) (2.08971,0.434714) };
          \addlegendentry{cuckoo};
          \addplot coordinates { (1.56442,1.71626) (1.57174,1.87607) (1.57403,1.8197) (1.58579,2.25305) (1.58731,2.00072) (1.59051,1.45768) (1.592,1.57137) (1.59556,1.58365) (1.60565,1.61816) (1.61414,2.66906) (1.61627,2.67418) (1.61909,2.56487) (1.62686,1.68237) (1.62857,1.7864) (1.63193,1.78584) (1.64341,1.68779) (1.66389,3.05672) (1.66695,3.1129) (1.66967,3.22774) (1.67839,3.0128) (1.69545,2.05342) (1.69689,2.10737) (1.70004,2.20268) (1.7093,1.43495) (1.72495,1.1391) (1.72789,1.24027) (1.73109,1.31496) (1.74172,1.38813) (1.79052,1.03655) (1.79199,1.12476) (1.79579,1.22265) (1.80519,1.38296) (1.87443,1.20399) (1.93918,1.04885) (2.00036,0.764025) (2.01172,1.03336) (2.06839,1.04912) (2.09022,1.03858) (2.14599,1.07173) };
          \addlegendentry{rotations};

          \legend{};
        \end{axis}
    \end{tikzpicture}
    \hfill
    \begin{tikzpicture}[baseline=-2cm]
        \ref*{paretoLeafMethodsShockHashLegend}
    \end{tikzpicture}
    \caption{Pareto front for space usage versus construction performance of different methods of finding a bijection.}
    \label{fig:leafMethodsShockHash}
\end{figure}
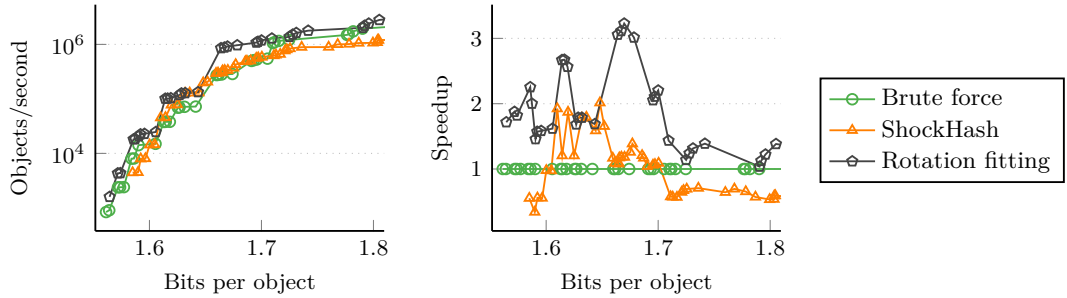

\Cref{fig:leafMethodsShockHash} gives experiments of our prototypical implementation of ShockHash, integrated as a bijection search method into RecSplit.
For the same leaf size, the space usage of the final MPHF is still larger than for the brute force approach because of higher constant factors.
However, since ShockHash is faster, we can also achieve larger leaf sizes.
In order to adapt to the faster bijection search, we modify the fanout to be 2 for the entire splitting tree.
The look at the Pareto front indicates that ShockHash is a practical solution.
Unfortunately, in its current state, it is unable to consistently beat even the brute force technique.
In case that the hash function evaluations are expensive, this can still be useful.
Another problem of ShockHash is that it is hard to parallelize the construction of cuckoo hash tables using SIMD instructions.

We have some ideas on how to improve the ShockHash construction.
An efficient implementation, as well as a formal proof for the presumption that $p\rightarrow 2^{-0.44m}$ is left for future work.

\end{document}